\documentclass{sig-alternate}

\usepackage{url,graphicx,amsfonts,amsmath,bm,accents}
\usepackage{times}
\usepackage{multirow}

\newtheorem{theorem}{Theorem}
\newtheorem{defn}{Definition}

\newcommand*{\dt}[1]{%
  \accentset{\mbox{\large\bfseries .}}{#1}}

\newcommand{\mN}{\mathcal{N}}
\newcommand{\mS}{\mathcal{S}}

\newcommand{\mC}{\mathcal{C}}
\newcommand{\mT}{\mathcal{T}}

\newcommand{\Esym}{\mathrm{E}}
\newcommand{\E}[1]{\Esym\left[#1\right]}

\newcommand{\bX}{\mathbf{X}}
\newcommand{\bx}{\mathbf{x}}

\newcommand{\bD}{\mathbf{D}}
\newcommand{\bp}{\mathbf{p}}
\newcommand{\be}{\mathbf{e}}
\newcommand{\bJ}{\mathbf{J}}

\begin{document}

%\title{Monitoring and Aggregation for Managing Dynamic Virtual Networks}

\title{TARDIS: Stably shifting traffic in space and time}

\numberofauthors{6}
\author {
\alignauthor Richard G. Clegg,\\
    \affaddr Dept of Elec. Eng. \\
    \affaddr University College London \\
    \email{richard@richardclegg.org}\\
\alignauthor Raul Landa\\
    \affaddr Dept of Elec. Eng. \\
    \affaddr University College London \\
    \email{raul.landa@ucl.ac.uk}\\
\alignauthor Jo\~{a}o Taveira Ara\'{u}jo\\
    \affaddr Dept of Elec. Eng. \\
    \affaddr University College London \\
    \email{j.araujo@ucl.ac.uk}\\
\and
\alignauthor Eleni Mykoniati\\
    \affaddr Dept of Elec. Eng. \\
    \affaddr University College London \\
    \email{e.mykoniati@ucl.ac.uk}\\
\alignauthor David Griffin \\
    \affaddr Dept of Elec. Eng. \\
    \affaddr University College London \\
    \email{dgriffin@ee.ucl.ac.uk}\\
\alignauthor Miguel Rio\\
     \affaddr Dept of Elec. Eng. \\
    \affaddr University College London \\
   \email{miguel.rio@ucl.ac.uk}\\
}
\maketitle

\begin{abstract}
This paper describes TARDIS 
(Traffic Assignment and Retiming Dynamics with Inherent Stability) which
is an algorithmic procedure designed to reallocate traffic within Internet
Service Provider (ISP) networks.  Recent work has 
investigated the idea of shifting traffic in time (from peak to off-peak)
or in space (by using different links).  This work gives a unified
scheme for both time and space shifting to reduce costs.  Particular
attention is given to the commonly used 95th percentile pricing scheme.

The work has three main innovations: firstly, introducing
the Shapley Gradient, a way of comparing
traffic pricing between different links at different times of day; 
secondly, a unified way of
reallocating traffic in time and/or in space; 
thirdly, a continuous approximation 
to this system is proved to be stable.
A trace-driven investigation using data from two service 
providers shows that the algorithm 
can create large savings in transit costs even when only small
proportions of the traffic can be shifted.
\end{abstract}

\section{Introduction}
\label{sec:intro}
Internet Service Providers (ISPs) that are predominantly used by 
residential
users (sometimes called eyeball ISPs) typically have traffic patterns
which are dominated by incoming traffic as their typical user downloads
more than they upload.  Managing
this traffic to reduce cost, network congestion and network instability is a
primary concern of such network operators. 
Traditionally, networks have attempted to manage demand through a 
combination of traffic shaping, artificially curbing demand, 
and traffic engineering through routing optimisations.
Some recent research has considered alternative solutions, moving
the incoming traffic in space (by downloading content from 
different physical locations) 
\cite{xie2008p4p,choffnes2008taming,frank2012cate} or in 
time (by shifting delay-tolerant traffic
to the off-peak) \cite{laoutaris2008good,wong2011time,wong2011time2}.  
Both temporal and spatial traffic shifting share the same underlying 
premise: that reallocating traffic can improve network performance
(by reducing costs, increasing stability or other goals).  This
work all deals with redistributing the traffic into and out of
an ISPs network.  However, usually the assumptions are simply to move
the traffic to a ``cheaper link" or to the ``off peak".  In fact
the trade offs may not be so simple and moving too much traffic
may worsen the situation.  This paper presents a procedure for
pricing the times and locations and an algorithm which shows how 
this price can be used to redistribute traffic in a stable way.
A continuous time approximation of the algorithm is provably stable.

Spatial shifting of traffic is studied in a number of contexts.  It
is often the case that content can be downloaded from different physical
locations.  In some hosting infrastructures as much as 93\% of content 
is hosted in multiple locations and by one estimate 40\% of
traffic could be downloaded from three or more locations \cite{ager2011}. 
Systems in common use which replicate content
across locations include peer-to-peer (P2P)
systems, content distribution networks (CDNs) and one-click hosting
services (OCH). 
In these systems a number of methods have been 
proposed or demonstrated which show that these extra copies can be
exploited to reduce traffic costs or for other engineering goals
\cite{xie2008p4p,choffnes2008taming,frank2012cate}.  Even when content
is available from only a single source then spatial shifting of
traffic is possible by using alternative routes 
\cite{goldenberg2004multihome} in a multi-homed network.

In parallel, multiple papers have explored the potential for
shifting delay-tolerant traffic to off-peak hours.
In \cite{laoutaris2008good} the authors describe a mechanism that offers
users higher bandwidth off-peak if they deliberately delay some of their
traffic.
Further contributions \cite{wong2011time,wong2011time2,chhabra2010home}
represent similar attempts to shift traffic in time through user
incentive schemes.

The papers above (and others) present a number of different alternative 
means to move traffic in either space and time and it seems certain
more will arise in other contexts (for
example, content centric networking explicitly encourages content to be 
available from multiple sources).  This paper presents
a control algorithm which could be used with any of the above systems
alone or in parallel.  The goal presented in this paper is traffic 
cost reduction but other engineering aims could be brought in as well
for example avoiding the onset of congestion on a link.

The contributions of this paper are threefold.
Firstly, the Shapley gradient is introduced, a means 
to compare the costs associated with traffic
flows from different sources at different times and subject to
different pricing schemes.  This is a general mechanism which, while
focused on the 95th percentile common in transit pricing, can be
used for many cost models common for ISPs.
Secondly, a unified mathematical framework is presented for reallocating
traffic across both time and space.  The algorithm shows how traffic
allocation should respond to prices by shifting traffic.  
Thirdly this reallocation strategy is shown to be stable.
The dynamical system representation of this mechanism is shown to converge
to a beneficial state for the system under weak conditions.
The properties of TARDIS are verified through analysis of real traffic
traces from a large European ISP and a Japanese academic
network.

The structure of the paper is as follows.
Section \ref{sec:background} gives background and related work in the area.
Section \ref{sec:pricing} describes an algorithm which can translate groups
of links and pricing schemes to comparable costs for putting traffic
on the network from a given source at a given time.
Section \ref{sec:dynamics} creates a dynamical systems model of the
traffic shifting in response to the calculated costs.
Section \ref{sec:modelling} describes a software model of the system 
and section \ref{sec:results} gives an evaluation using
real user traffic data.
Finally section \ref{sec:conclusions} gives conclusions about the TARDIS
system.

\subsection{A brief description of the TARDIS procedure}

This section briefly describes the TARDIS procedure as a whole.
The focus of the procedure is to control where and when traffic
should be assigned as it arrives on a network.  If a requested
download could be made from a different location or if it could
be shifted in time (by several hours) to the off-peak then the
TARDIS procedure calculates at what time and location the
traffic should be placed in order to minimise costs.
The mechanisms by which this reassignment could occur
are dealt with by many other papers in the literature (as
detailed in section \ref{sec:background}) and several different
mechanisms are proposed. The
exact details of the mechanism used for reassignment do not affect
the TARDIS procedure.  

\begin{figure}
\includegraphics[width=8.5cm]{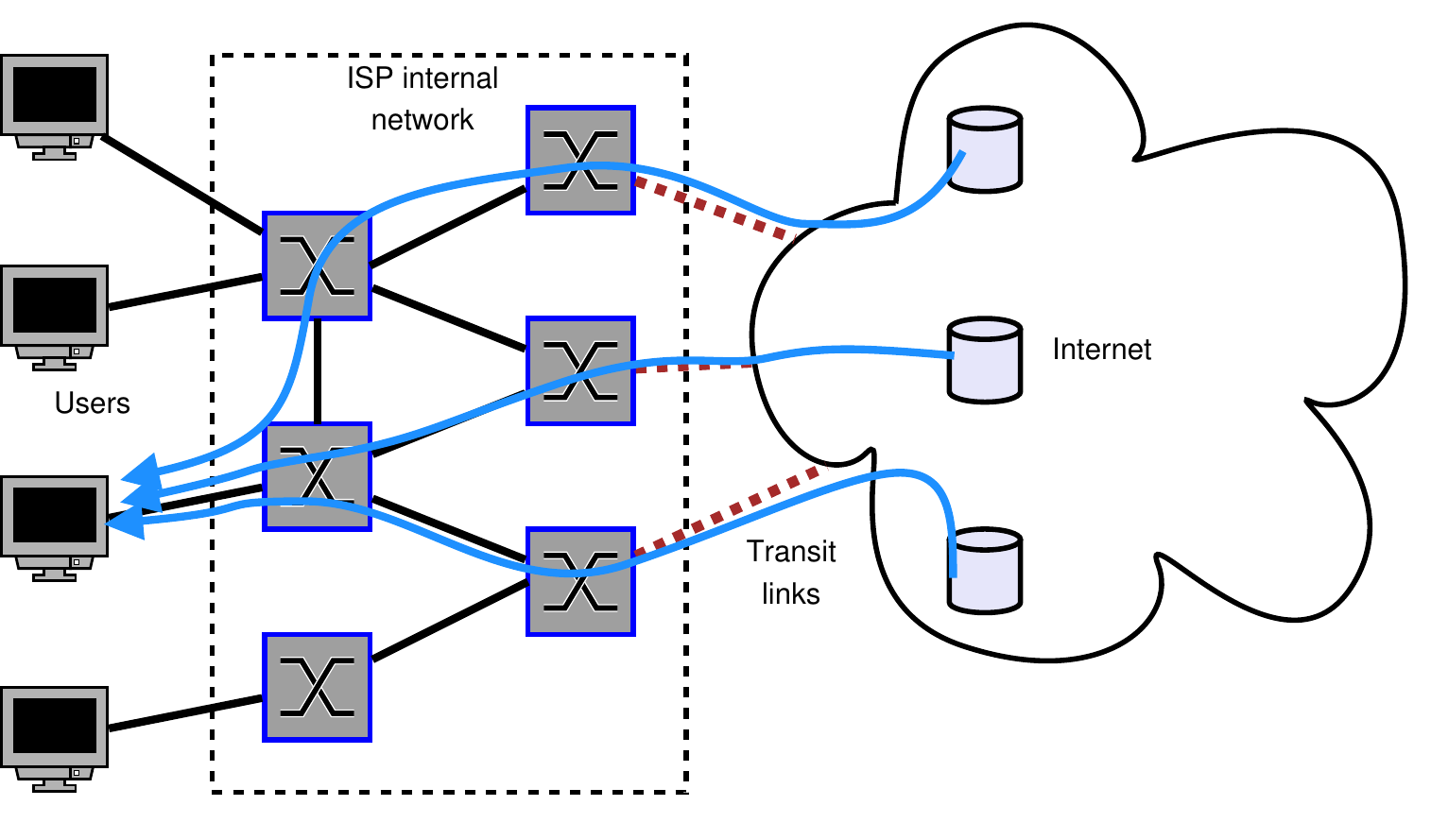}
\caption{An eyeball ISP network where a user can download requested
content from three different transit links.}
\label{fig:network}
\end{figure}

Figure \ref{fig:network} shows the situation envisaged 
where shifting in space is possible.  In this case, the user wishes
to download content which is not available from the ISP's internal
network.  The content is available from three separate transit links
that are charged at different rates.  The decision as to which 
transit link to use is not trivial as it is not always possible to
know which link is ``cheapest" at a given time.
This is discussed in depth in section \ref{sec:pricing}.
Note that, of course, some links may be peering links not transit
links, this does not affect TARDIS if the pricing model is one
which can be addressed by the TARDIS system.  
Consideration to these pricing schemes is
given in sections \ref{sec:fixed_price} and \ref{sec:other_price}.

TARDIS is given as input a cost
model (details of how the ISP is to be charged for its traffic)
and the available choices for a given piece of traffic
(the times and locations to which that traffic could be assigned, which
of course, includes the possibility it cannot move at all).
The TARDIS model then computes splitting rates which describe
the desired proportion of traffic from the choice set which
should be assigned to each time and location.  
Note that obviously
individual flows are not split up to different destinations and
times.

\begin{figure}
\includegraphics[width=8.5cm]{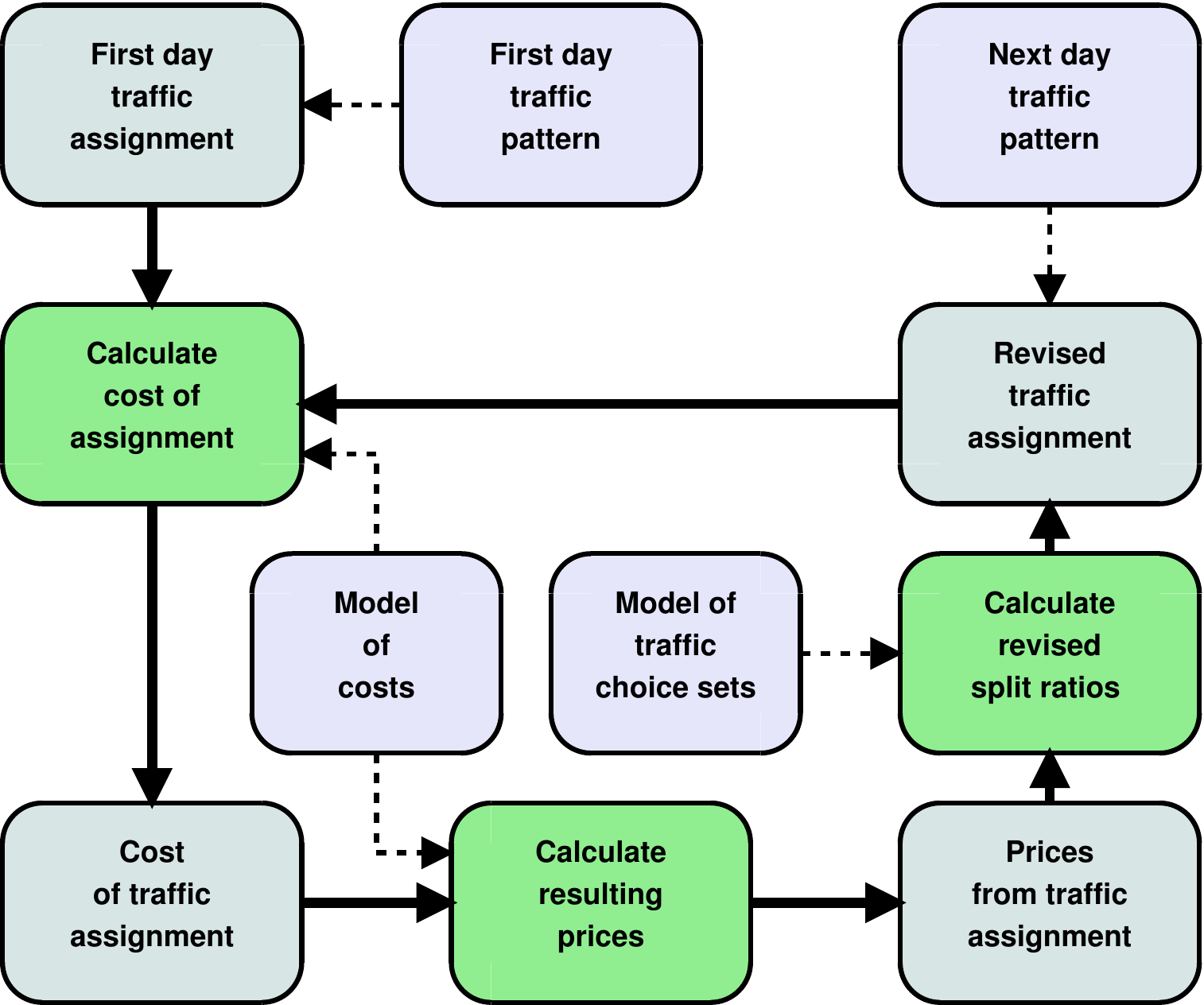}
\caption{The TARDIS scheme for shifting traffic in space and time
to reduce costs.}
\label{fig:scheme}
\end{figure}

Figure \ref{fig:scheme} shows the basic iterative loop which would
be used by TARDIS.  The dotted lines show input, the solid lines
flow of control.  On the first day, traffic is assigned without
modification.  The resulting costs arising 
from this assignment are calculated from a model of how the ISP is
charged for their traffic.  These costs are used to calculate
virtual prices for traffic which is assigned to a given time
and location.  The mechanism for this calculation is described
in section \ref{sec:pricing}.  These prices are used to calculate
splitting ratios which could be interpreted as the probabilities of
traffic choosing between various destinations (or times) picking
a given destination.  The mechanism for calculating the splitting
ratios is given in section \ref{sec:dynamics}.
These splitting ratios will, in combination
with the traffic placed on the network create a new cost.  The new
cost will create new prices and so on.  Obviously it is of primary
importance that this system is stable and this is proved in
section \ref{sec:dynamics} for a continuous time approximation of
the system.

\section{Background}
\label{sec:background}
ISPs pay for traffic on their networks in a variety of ways
either directly through charges levied based upon the traffic level
or indirectly by paying for the cost
of leasing appropriate network capacity.
In transit networks the most common way to charge
for traffic is 95th percentile pricing
(see section \ref{sec:percentile}).  
If 95th percentile pricing is applied to traffic generated by
many users, it is natural to 
ask how much each user contributes to the total 
cost and how this could be mapped to a cost per time period.
One answer to this question is given
in \cite{stanojevic2010heavy}, where the authors use 
the Shapley value to find user cost contributions and least-squares fitting to find
the unit costs that best approximate the Shapley value as a function
of the traffic volume and time of day.  This is described more fully
in section \ref{sec:percentile}.  In \cite{valancius2008howmany}
the authors consider pricing strategies for transit ISPs noting that transit
traffic prices vary according to destination.  They find that the strategy
of discounting local traffic is not optimal and instead suggest an automatic
way to bundle traffic into a small number of pricing tiers according to cost
and demand.  In \cite{lakhina_cost_2012} the authors consider the
benefits of ISPs changing routing decisions to move traffic around their
internal network to reduce cost.  They account for a number of different
network costs including fixed costs, interconnect costs, transit and 
backhaul.  They produce an optimisation model which outputs routing
decisions and uses the Shapley value formulation from 
\cite{stanojevic2010heavy} to assess traffic where 95th percentile
billing is used.  The work is considerably different in scope and
intent to this paper.  Their aim is to produce a comprehensive
cost model and demonstrate that routing decisions in the internal
network could save costs given a fixed traffic matrix.  This paper 
takes a simpler cost model
and focuses on the strategies which could be used to save costs
by altering the traffic matrix.  In fact, the cost
models from \cite{lakhina_cost_2012} could be used as an input
to the TARDIS system to improve its ability to shift traffic on
the internal network as well as external traffic.  The costs
models formulated in \cite{lakhina_cost_2012} could all be
translated simply to the Shapley gradient formulation in section
\ref{sec:pricing}.

A number of authors have investigated how eyeball ISPs can reduce their
traffic costs by incentivising users to delay downloads.
In \cite{wong2011time} and \cite{wong2011time2}, the authors describe a
model which uses a control loop to adapt the prices that ISPs charge users
in response to changes in their bandwidth consumption.
This provides an incentive for users to shift part of their traffic to
off-peak times.
Another deferred download scheme is the Internet Post Office 
\cite{laoutaris2008good} in
which users request files and the ISP downloads them off-peak and temporarily stores
them so that users can quickly retrieve the local copies
when they next log on.
The idea is further developed in \cite{chhabra2010home} which uses real
user data to estimate the cost reductions provided by such time
delays.

Spatial shifting of external traffic is studied in various contexts.
Routing via different transit links in a multi-homed network (to/from
the same destination) is studied in
\cite{goldenberg2004multihome} which solves the problem of rerouting
to reduce network costs either for
percentile based pricing or linear cost pricing but not a mixture.
In \cite{CooperativeISPCDN} the authors investigate cooperation between
content providers and ISPs for traffic engineering and to improve
server selection by choosing which replica to download from.  
Their system shows several benefits from merely information sharing about
connections and larger benefits from more active cooperation in choice
of physical location for connections.  

For peer-to-peer systems a number of systems exist. 
ALTO/P4P \cite{xie2008p4p} was tested in simulation and experiment and 
shown to reduce network transit costs.  ONO \cite{choffnes2008taming} which was
deployed and reduced inter ISP traffic by changing peer selection.
The ALTO/P4P approach is to introduce an interface
between ISPs and overlay applications with the purpose
of facilitating the selection of overlay nodes based on locality.
In \cite{rumin2011deep} the authors assess the extent to which BitTorrent
swarms can be \emph{localised}, i.e. downloads can be kept within an ISPs
own network, reducing the ISPs transit costs.
The authors consider various strategies to bias overlay
topology construction towards local peers, and develop the concept of
inherent localisability, which assesses the download performance of
swarms using largely local connections within an ISP.
Unfortunately, the degree of localisability depends heavily both on the
nature of the torrent and the ISP.

Opportunities for space shifting are also widely recognised within 
Content Distribution Networks (CDNs),
which present both high content replication and transparency in traffic redirection.
In \cite{poese2010improving}, for example, alterations are made to DNS
servers in order to serve traffic from different CDN hosts transparently to
the end user.
In \cite{frank2012cate} the authors propose content aware traffic
engineering (CaTE), which allows ISPs to take advantage
of content available in multiple locations to reduce link utilisation.
The authors show that the gains can be substantial: more than 32\% of the
traffic in their dataset can be be delivered from at least 8 different
subnets, and almost 40\% of traffic can be obtained from 3 or more
locations.  This estimate focuses on traffic from major providers and does
not include, for example, P2P.

One-click hosting services could provide another opportunity for
spatially shifting traffic as they contribute a large proportion
to traffic share and many are multi-homed \cite{antoniades2009}.

The proposed TARDIS algorithm could be used in combination with any
of the systems from \cite{poese2010improving,frank2012cate,xie2008p4p,choffnes2008taming,wong2011time,wong2011time2,lakhina_cost_2012,laoutaris2008good}
as a control system to suggest which reallocations would most benefit 
the ISP.

\section{Definitions}
\label{sec:definitions}

Some definitions are made here to simplify further discussion and to
gather terms for the rest of the paper.

Define a \emph{traffic pricing group} (TPG) as traffic grouped 
together for pricing purposes.
In the simplest case this may be traffic over a single physical link.
It may also represent several links where the traffic is aggregated to
produce a final price
or even a subset of traffic on a single link priced differently.
The latter arises for example when providers charge for
national and international traffic at different rates,
a practice often referred to as unbundling \cite{valancius2008howmany}.

Define a \emph{traffic slot} as a specific TPG associated with a
specific time window.
Hence, a set of traffic slots represents a choice of TPGs and time windows
to which traffic volumes could be allocated.
For example, one set of slots may represent the possibility of downloading
from a fixed location at any time in the next eight hours;  another may
represent the possibility of downloading only during the current time
window, but from many different TPGs.

The main variables used in the next sections can be found in the following
table.

\begin{center}
\footnotesize
\begin{tabular}{c|p{7cm}}
$v(\mS)$ & Total cost of a TPG for traffic from users in set $\mS$\\
$p_i$ & price of using slot $i$\\
$\mC_j$ & Set of slots making up the $j$th choice set\\
$d_i$ & Traffic demand in bytes for traffic which could be assigned to any
slot in $\mC_i$ \\
$X_{ij}$ & Traffic flow in bytes which could be assigned to any slot in
$\mC_i$ and is assigned to slot $j$\\
$s_{ij}$ & the proportion of traffic which could be assigned to any slot in
$\mC_i$ and is assigned to slot $j$
\end{tabular}
\end{center}

The following notational conventions are used throughout this paper.
Lower case bold (e.g. $\mathbf{x}$) indicates vectors.
Calligraphic script indicates sets (e.g. $\mathcal{S})$.
Upper case bold indicates a matrix or a vector of sets (e.g.
$\mathbf{X}$).

\section{Pricing}
\label{sec:pricing}
ISPs typically pay a number of different costs for traffic entering
and leaving their network their network.  These may include fixed 
costs such as provisioning a link of fixed capacity on their 
internal network, or to an Internet eXchange Point (IXP) and 
variable costs, for example a transit
link which is charged according to the volume of traffic used.
Of particular importance is the 95th percentile pricing scheme,
described in section \ref{sec:percentile}.  The nature of
the 95th percentile pricing scheme makes it difficult to
deal with analytically (as will be explained in the next section).  
In previous work \cite{stanojevic2010heavy} 
the Shapley value has been used to estimate the price of a particular
time and link.  

Through this section the words ``price" and ``cost" will be used with
very specific meanings.  The cost of a traffic pattern is an outcome of
the shape of the traffic and the pricing scheme imposed -- it is the
monetary value which would actually be paid for that traffic using
that scheme.  The price
will be used in this section to mean a notional marker for a given
traffic slot which indicates the likely cost impact of assigning traffic
to that slot.  This price is used as an internal mechanism to work out
which traffic slots should be avoided.  If small amounts of traffic are
moved from a high price slot to a low price slot then it would be
expected that the cost would drop.  

The following properties are useful for the price  
chosen:
\begin{enumerate}
\item The price can be quickly calculated.
\item An increase of traffic in a slot never causes the price of
that slot to fall (monotonic with traffic).
\item The price is differentiable with respect to traffic.
\end{enumerate}
The first condition is for practicality.  The second condition simply
says adding traffic never makes the price go down.
The second and third conditions are used in the stability proof in section 
\ref{sec:dynamics}.

In this paper the Shapley
gradient is introduced to solve this problem.  
This can be considered to answer the question
``what is the likely increase in cost which would be caused by adding
traffic to this link?" in a robust way which can account for a wide
variety of pricing schemes
including, naturally, the 95th percentile.

\subsection{95th percentile pricing in brief}
\label{sec:percentile}

For transit traffic, ISPs are commonly charged for
the 95th percentile of their traffic.  This works as follows.
For each TPG
the 95th percentile cost is set to some fixed value
e.g. $r$ dollars per GBps (note that this is a rate not an
absolute value). For the charging
period $T$ (a typical
value is a month)
the traffic is divided into smaller  
time windows of length $t$ (often 5 minutes is used).
For each window $i$ the average traffic rate $f_i$ is calculated 
(it is the total traffic in that window divided by $t$ the window
length).
Define $f^{(95)}$ to
be the traffic rate $f_i$ such that only $5\%$ of the 
$f_i$ are larger than $f^{(95)}$.
The price charged for the period $T$ is then simply $r f^{(95)}$.
It is usual that inbound and outbound traffic are tracked
separately, and only the largest charged (see \cite{dimitr2009perc}).
For eyeball ISPs this will almost always be the inbound traffic as this
is larger in volume than the outbound traffic.
The 95th percentile pricing provides a particular challenge for any scheme 
which aims to reassign traffic.  In particular the question ``how
does adding traffic to this slot affect the amount paid?" becomes
problematic.  Adding or taking away a small amount of traffic to any
slot has no affect on the cost unless that slot is one with 
traffic level $f^{(95)}$ for that TPG.  A more subtle analysis is
required and this is provided by building on the Shapley value as
studied in \cite{stanojevic2010heavy}.  

\subsection{Calculating prices using the Shapley gradient}
\label{sec:shapleygradient}

Consider the traffic in a single TPG with a known pricing scheme and
with $N$ users.
Let
$v(\mS)$ be the total cost which would be paid using
this scheme for the traffic 
generated by a some set of users $\mS$.  Define $\mN$ as the
set of all $N$ users.
The Shapley value is a concept from game theory which assesses the
contribution of a user's strategy to an overall cost/benefit.  The
Shapley value 
(see \cite{stanojevic2010heavy}) of
the $i$th user is defined as 
\begin{equation}
\phi_i = \frac{1}{N!} \sum_{\pi \in \mS_{\mN}} [v(\mS(\pi,i)) - 
v(\mS(\pi,i) \backslash i)] \mbox{,}
\label{eqn:shapley}
\end{equation}
where $\mS_{\mN}$ is the set of all $N!$ possible permutations of $\mN$,
$\pi$ is one such permutation, and $\mS(\pi,i)$ is the set of users who 
arrive not later than $i$ in the permutation $\pi$. 
Intuitively, \eqref{eqn:shapley} can be interpreted as randomising the order of users,
estimating the cost incurred by each user $i$ and averaging this cost over all possible 
user orderings. 
In \cite{stanojevic2010heavy} the Shapley value of a user's traffic
is used to assign a cost to each hour of the day which reflects the possibility
of traffic in that slot contributing to an increased price.
The full calculation of the Shapley
value \eqref{eqn:shapley} requires considering $N!$ combinations
of user traffic to analyse traffic from $N$ users but
\cite{stanojevic2010heavy} shows that a relatively 
small sampling gives
an efficient, unbiased estimator (in their work 1,000 orderings produced
a low error estimate) and this sampling technique is used here.

In \cite{stanojevic2010heavy} this value (that differs for every user)
is combined with a least squares fit to get an average cost to assign
to each hour.  However, this procedure is computationally intensive 
(calculating the Shapley value for every user and then doing a least
squares fit) and does not produce a cost which can be compared 
to other pricing schemes.
What is required is some measure of the cost of adding a small amount
of traffic on a given TPG at a given time.  This is achieved by
considering the gradient of the Shapley value.

Define the \emph{Shapley gradient} as the rate of
increase in cost when a fictitious user $N + 1$
injects an additional small amount of traffic $du$ in slot $j$.  
The Shapley gradient is therefore,  
\begin{align}
\nonumber G_j  &=  
\frac{1}{du\,(N+1)!} \\ 
& \sum_{\pi \in \mS_{\mN'}} [v(\mS(\pi,N+1)) - 
v(\mS(\pi,N+1) \backslash N+1) ],
\label{eqn:Gjgeneral}
\end{align}
where 
$\mS_{\mN'}$ is the set of arrangements of the
users $\mN$ plus the fictitious $(N+1)$th user and 
$\mS(\pi,N+1)$ is the set of all users arriving
not later than user $N+1$ in the permutation $\pi$.

Define the Shapley gradient for an individual user $i$ and
slot $j$ as 
$\phi'_{ij}= (\phi_i(d_{ij}) - \phi_i) / du$ where
$\phi_i(d_{ij})$ is the Shapley value for user $i$ with
extra traffic $du$ in slot $j$. For the traffic schemes considered
here this can be shown to be approximately
the mean of $\phi'_{ij}$ over all users $i$ with 
the error term O(1/N).  In
all but the 95th percentile case $\phi'_{ij} = G_j$ for all
$i$.  Details are given in appendix \ref{sec:shap_indep}.

\subsection{Costs for various pricing schemes}

In this section the Shapley gradient is calculated for various
pricing schemes.  The Shapley gradient is a quite general
concept and works for any scheme where the Shapley value
is differentiable.  This condition amounts to saying that the
pricing scheme is such that there is no step change with
induced traffic.  This would not be the case with, for example,
a scheme which charged a fixed rate up to a given amount of
traffic and then a higher rate above that amount.
In the section \ref{sec:dynamics} it will also be useful
that schemes are differnentiable with respect to added traffic.

\subsubsection{Shapley gradient for linear pricing}
\label{sec:linear_price}

For traffic in a slot $j$ charged with 
linear pricing at rate $r_j$ then 
\eqref{eqn:Gjgeneral} reduces to simply $G_j = r_j$ as
would be expected.  
The equation becomes 
$G_j = \frac{1}{du \, (N+1)!} 
\sum_{\pi \in \mS_{\mN'}}  r_j du$
since the cost of adding $du$ traffic to a slot $j$ which
is priced linearly at $r_j$ is $r_j \, du$ and hence
$v(\mS(\pi,N+1)) - 
v(\mS(\pi,N+1) \backslash N+1) = r_j \, du$.
Since there are
$(N+1)!$ members of $\mS_{\mN}$ the equation reduces
to $G_j = r_j$.  

\subsubsection{Shapley gradient for 95th percentile pricing}
\label{sec:shap95}

Let $\mT^{(95)}(\mS'(\pi))$ be the set of all
time periods which have a value equal to the 95th percentile value
of the traffic up to and including user $N+1$ in arrangement $\pi$.
Let $I(X)$ be an indicator variable -- that is a variable which takes
the value $1$ if X is true and $0$ if X is false.
Consider a slot $j$ charged at
95th percentile at rate $r_j$.  
For some arrangement of traffic $\mS(\pi,N+1)$ then there are two
possibilities.  If, for the traffic profile $\mS(\pi,N+1)$, then 
the flow in slot $j$ is equal to $f^{(95)}$ then adding $du$ to
slot $j$ increased the cost by $r_j \, du$.  In all other cases
then adding $du$ did not increase the cost.  Therefore,
$v(\mS(\pi,N+1)) - 
v(\mS(\pi,N+1) \backslash N+1) = I (j \in \mT^{(95)}(\mS(\pi)) r_j \, du$.
Intuitively this says that adding traffic $du$ to slot $j$ in
an arrangement of traffic increases the 95th percentile 
cost excactly when the traffic has been added to the 95th percentile
slot and at no other times.
Hence,  
\eqref{eqn:Gjgeneral} can be rewritten as
$$
G_j = \frac{1}{du\, (N+1)!}  
\sum_{\pi \in \mS_{\mN'}} [I (j \in \mT^{(95)}(\mS'(\pi)) r_j \, du ].
$$ 
This then gives
\begin{equation}
G_j = \frac{r_j}{ (N+1)!} \sum_{\pi \in \mS_{\mN'}} 
[I (j \in \mT^{(95)}(\mS'(\pi))) ].
\label{eqn:Gj}
\end{equation}
As in \cite{stanojevic2010heavy}, only a small sample
of all combinations in $\mS_{\mN'}$ need be calculated and
this can be computed efficiently see \cite{stanojevic2010heavy}
and \cite{lakhina_cost_2012} for further information on the
performance of the estimation.

The expression \eqref{eqn:Gj} has the following properties useful to construct the 
slot prices $p_j$.
\begin{itemize}
\item When summed over all time windows associated with a TPG, 
the result is the 95th percentile price actually paid for that TPG
with that traffic.
\item It reflects the likelihood of adding traffic in a given time
slot increasing the 95th percentile cost.
\item Off-peak slots have near zero $G_j$.
\end{itemize}

\subsubsection{A useful approximation for 95th percentile pricing}

A problem remains with the 95th percentile $G_j$ as defined in \eqref{eqn:Gj}.  
Moving traffic to a very busy period is just as valid a strategy for 
cost reduction as moving it to a quiet period.  
This follows because shifting traffic from busy time periods to quiet
ones will lead to a lower 95th percentile, and hence, to a reduction in costs.  
However, shifting traffic to the busiest time periods, so that it 
falls in the top 5\%, could also reduce costs.
In practice such a policy would impact end-to-end performance adversely.
The price function is therefore modified by changing $\mT^{(95)}$ 
in \eqref{eqn:Gj} to $\mT^{(>95)}$, the set of slots with a traffic 
level equal to or greater than the 95th percentile level and 
normalisting by the size of this set.
$$
G_j = \frac{r_j}{ (N+1)!} \sum_{\pi \in \mS_{\mN'}} 
\left[\frac{I (j \in \mT^{(>95)}(\mS'(\pi)))}
{|\mT^{(>95)}|}
\right].
$$

This alteration has the added benefit of making the slot price $p_j$ 
monotone as a function of the traffic in that slot.  The results
for TARDIS use this as the price function but use the unmodified 
95th percentile as the cost function.  This has the benefit of not
inducing unrealistic traffic profiles although it means that cost
reduction is not sought as aggressively as it might be.

It will later be a useful property that the price
of a slot is monotonically non decreasing but also that it is continuously 
differentiable.  This can be achieved by constructing the following
approximation 
\begin{equation}
G'_j = \frac{r_j}{ (N+1)!} \sum_{\pi \in \mS_{\mN'}}
\frac{A(j,\pi)}{\sum_{k} A(k,\pi)},
\label{eqn:Gjdash}
\end{equation}
where the $k$ sum is over all time windows and
$A(j,\pi) = 1$
if $t(j,\pi) > f^{(95)}(\pi)$ and
$\exp[(-(t(j,\pi)- f^{(95)}(\pi))^2/\sigma^2]$ otherwise.
Here $t(j,\pi)$ is the traffic level in slot $j$ counting traffic up to the $N+1$th user in arrangement $\pi$ and $f^{(95)}(\pi)$ is the 95th percentile level for the traffic up to this user.
The parameter $\sigma \geq 0$ is akin to variance.  This effectively fits a Gaussian shape to the price for traffic below the 95th percentile level.  The fall off is controlled by $\sigma$ and as $\sigma \rightarrow 0$ the approximation 
to the previous formula becomes exact. 

\subsubsection{Shapley gradient for fixed pricing with a maximum bandwidth}
\label{sec:fixed_price}

A common cost model on links is to pay a fixed cost for bandwidth
up to a given cap.  This could occur, for example, if the ISP pays
for a link to an IXP with a given rate.  Another situation where this
would occur is an internal link within the ISP which has fixed capacity
and must not be overloaded.  As the fixed cost is already paid, the
price for putting traffic on the link is zero as long as the traffic
remains below the cap.

A naive approach to this cost system presents a problem for the
TARDIS system as the cost would be zero up to the capacity
then an infinite cost at that capacity when
the link fails (more realistically, the cost would approach
infinity, link failure, as the link approaches its maximum utilisation). 
Fortunately, a number of 
pricing schemes are possible which allow modelling an approximation
to this cost function.  The following scheme is inspired by the
well-known result from queueing theory that the mean queue length
for an M/M/1 queue with utilisation $\rho \in (0,1)$ is given by
$\E{N} = \rho/(1-\rho)$.

Let $m$ be the maximum traffic rate allowed on a link and 
$\alpha \in (0,1)$ be some
proportion of that rate which can always be tolerated (for example
if $\alpha = 0.8$ the link is considered unpriced until
its utilisation is 80\%).  The cost for
a slot $j$ carrying flow $f_j$ could be approximated by
$$
c_j = 
\begin{cases}
0 & 0 \leq f_j < \alpha m \\
\frac{f_j-\alpha m}{m-f_j} &  \alpha m \leq f_j < m.
\end{cases}
$$
This function gives a cost 0 up to $\alpha m$ and then rising
until the cost approaches infinity as the traffic in the
slot approaches $m$.
The Shapley gradient will simply be the differential of this.
If $f_j$ is the flow on slot $j$ then the price of that slot
is given by the Shapley gradient
$$
G_j =
\begin{cases}
0 & 0 \leq f_j < \alpha m \\
\frac{(1-\alpha)m}{(m-f_j)^2} &  \alpha m \leq f_j < m.
\end{cases}
$$
This is equivalent to a price which is zero (or fixed and finite)
for traffic less than $\alpha m$ and rising rapidly to infinity as
$f_j$ approaches $m$.

\subsubsection{Shapley gradient for other pricing schemes}
\label{sec:other_price}

It is important to consider the case where more than one
cost constraint affects traffic.  For example, in figure 
\ref{fig:network} traffic must cross the ISP internal network
after the transit links.  The ISP internal network may have bandwidth
constraints on links which form an additional part of the problem in
addition to minimising the cost from external links.  In this case
the cost of using the external link could be modelled as the sum of
the cost from transit or peering plus a fixed price weight as in 
the previous section when the associated internal link is running 
near capacity.  If it was known that a certain transit link experienced
congestion at some times then this scheme could also be used to limit
such assignment.  The TARDIS system is relatively flexible in the
cost/pricing which can be assigned as long as the cost can be 
approximated by a differentiable non decreasing function of the
traffic.

The previous sections cover a large number of pricing schemes, however
one practice not yet dealt with is the idea of pricing bands.
For example, traffic might be charged at a particular rate (either 
linearly, fixed cost or at 95th percentile) up to a given level and then at a
lower rate beyond this level.  This presents a problem for the
TARDIS system in two ways: firstly, the cost is no longer non-decreasing,
at a certain point adding traffic reduces the cost; secondly these
pricing bands are often pre-agreed, it would not be realistic to have
an automated system switch between pricing tiers according to changing 
traffic patterns.  Therefore such a cost model would be dealt with
within TARDIS by assuming that the current price band is an input and
the decision to move up or down a price band is an externally made
engineering decision (made by humans or by expert systems) which can
be fed into the TARDIS system if the decision is made to change the 
price band.

\section{Dynamical systems approach}
\label{sec:dynamics}
Having assigned a price $p_i$ to each slot $i$, this section 
presents a mechanism to allocate traffic to each slot so that 
the total cost is reduced.  This is equivalent to changing
the time and/or TPG of a piece of traffic.  In
this section a 
dynamical system is formulated for which a solution with a stable 
equilibrium is provided.  This system is inspired by a
system developed in the context of road traffic 
\cite{smith1984stability}.

A naive approach would be to assign all traffic to the cheapest slot 
available.  This, however, could easily cause problems.  Shifting a 
significant amount of traffic in this manner may inflate slot price 
excessively,
potentially reducing traffic shifting to a small subset of slots 
with wildly fluctuating prices, in a situation similar to route flapping.

In addition to the potential for oscillation, the problem
is further complicated by the fact that not all traffic can 
be allocated to all slots.  Some (possibly large) proportion of 
the traffic can only be assigned to a single slot, being delay-intolerant and
having no alternative download locations. Other traffic may impose different 
restrictions, such as only being allocatable to slots in the same time window
(space shifted but not time shifted),
or to slots associated with the same location 
(time shifted but not space shifted).
The approach chosen in this paper is to define choice sets, the sets
of slots that a user's traffic can take.  This could be:
\begin{enumerate}
\item a single slot (for traffic which cannot be moved),
\item a set of different TPG at the current time (for traffic which can move in
space but not time),
\item a set including the current time and a number of subsequent times
up to a given maximum delay for a single TPG (for traffic which can
move in time but not space),
\item the cartesian product of 2) and 3) for traffic which can shift
in both space and time.
\end{enumerate}

\subsection{The stable assignment problem}
\label{sec:problem_definition}

Define a choice set $\mC$ as a set of slots to which a given unit
of traffic may be allocated.  
Let $\mC_i$ be the $i$th such choice set, 
and assume that each unit of traffic demand has an associated choice set
$\mC_i$.
Let $d_i$ be the amount of traffic within
one period which 
is free to choose among these choice sets and assume this is 
fixed. 

Let $X_{ij}$ be the customer traffic demand that can be 
allocated to slots in choice set $\mC_i$ and is in fact allocated 
to a slot $j$ (naturally, $X_{ij}= 0$ if 
$j \not \in \mC_i$).  Let $\bX$ be the matrix of all such $X_{ij}$.
For each slot $i$, the price $p_i$ is a function of the assigned traffic
(in 95th percentile the slot price depends on the assignment to
all slots in its pricing period).  Write this explicitly as
$p_i(\bX)$.  
\begin{defn}
Let $\mC_i$ be a choice set and $\bX$ be the assignment
of traffic within all choice sets.
The choice set is said to have an
\emph{equilibrium assignment} (or to be in equilibrium)
if $p_k(\bX) > p_j(\bX)$ 
implies $X_{ik} = 0$.  The system is said to be in
equilibrium if all choice sets are in equilibrium.
\label{defn:eqbm}
\end{defn}
In other words, at equilibrium, all the traffic is assigned
to one or more slots all of which have equal cheapest price.
This definition corresponds to the user 
equilibrium of road traffic analysis \cite{wardrop1952aspects}, 
or Wardrop equilibrium of the first type.
It embodies the idea that at equilibrium no traffic can be assigned
to a lower priced slot.

The \emph{stable assignment problem} can now be stated in general
terms: find a method
of allocating flows to slots within their choice set such that the prices
and flows resulting from such an assignment are an equilibrium assignment
for all choice sets.  Evidently, this represents an equilibrium
state because once it is reached, no traffic could be assigned to a 
different, lower price slot.

\subsection{A continuous dynamic model of slot choice}
\label{sec:basicmodel}

The model is now formulated in terms of assigning
traffic (choosing the $X_{ij}$) 
to achieve equilibrium by moving traffic within its choice
set to lower price slots.
The demand $d_i$ represents the amount of traffic which
can choose slots in $\mC_i$ and is $d_i = \sum_{j \in \mC_i} X_{ij}$.
Furthermore, since the $p_i$ will depend on the $X_{ij}$  
the problem of creating a stable dynamic system is then that of
creating an adjustment process for the $X_{ij}$ that moves
the system towards equilibrium. 
%A complicating factor is that
%the $c_i$ are not (usually) functions of the $f_i$, but of the $f_i$ disaggregated
%by user (when the Shapley value is accounted for).  
In practice both the $X_{ij}$ and the $p_i$ would be available only 
at the end of each pricing period $T$.
For analytic tractability, however, a continuous time formulation
is now given.

The proposed adjustment strategy is as follows.  If $p_k > p_j$, 
traffic which can be shifted from slot $k$ to slot $j$ will do so at a rate proportional
both to the price difference $p_k - p_j$ and to the amount of traffic 
$X_{ik}$ currently using slot $k$.  
Consider a toy example with two slots, 1 and 2, and a single 
choice set $\mC_1 = \{1,2\}$ (i.e. all traffic can choose between both slots).  
Assume some initial assignment $X_{1\,1}$ and $X_{1\,2}$.
If $p_1 > p_2$, then the rate of change of the two variables would be
$\dt{X}_{1\,1}= - X_{1\,1} (p_1 - p_2)$ and 
$\dt{X}_{1\,2}= X_{1\,1} (p_1 - p_2)$, where $\dt{X}$ indicates
time differentiation.  
This means that traffic will be re-allocated 
from slot 1 to slot 2: $X_{1\,1}$ will decrease
and $X_{1\,2}$ will increase, leading to 
an increase in $p_2$, a decrease in $p_1$ and a reduction in 
both $\dt{X}_{1\,1}$ and $\dt{X}_{1\,1}$ as the system approaches 
equilibrium. 
Since $X_{1\,1}$ decreases by as much as $X_{1\,2}$ 
increases, these equations are flow preserving;
in addition, it can be seen that 
if given positive initial values, 
neither $X_{1\,1}$ nor $X_{1\,2}$ will produce negative values. 

More generally, for all $j \in \mC_i$,
\begin{equation}
\dt{X}_{ij} = \sum_{k \in \mC_i} \left[ X_{ik}(p_k-p_j)_+ -
X_{ij}(p_j-p_k)_+\right],
\label{eqn:dynsys}
\end{equation}
where the notation $[x]_+$ means $\max(0,x)$.  The first expression
inside the sum shifts traffic towards $j$ if $p_k > p_j$,  and 
the second one shifts traffic away from $j$ in the opposite 
case.
If choice set $\mC_i$ has only one member, the sum is empty and 
$\dt{X}_{ij} =0$
as expected since $X_{ij}$ has no freedom to move. 

For reasons of space, only an outline of a stability 
proof can be given in this
paper; it proceeds as follows.  
The evolution of the assignments $X_{ij}$ can be described
with the differential matrix equation
$\dt{\bX} = 
= \Phi(\bX)$ which can be decomposed into equations for
each choice set $\dt{\bx_i} = \Phi(\bx_i)$ where
$\bx_i$ is the $i$th row from $\bX$ -- that is the
vector of flows in choice set $\mC_i$ that is 
$\bx_i = (X_{i\,1}, X_{i\,2}, \ldots)$ (noting that $X_{i\,j} = 0$
if $j \notin \mC_i$.
If the system is at a Wardrop equilibrium
then $\dt{\bX}(t) = \mathbf{0}$ 
since for any $\mC_i$ then in \eqref{eqn:dynsys} 
the term
$X_{ik}(p_k-p_j) = 0$ and $X_{ij}(p_j-p_k)= 0$ as either the price
term or the flow term is zero. 
Let $\bD$ be the set of $X_{ij}$ which are \emph{demand feasible}
(that is $\sum_j X_{ij} = d_i$ for all $i$) and assume that the
system starts in a demand feasible state.
The following theorem that is a variant of
\cite{smith1984stability} applies.
For set $\mC_i$ define $\bx_i$ as \emph{demand feasible} if
$\sum_j X_{ij} = d_i$ and (slightly abusing notation) $\bx_i \in
\bD$.  Define $\nabla_i(x)$ as the gradient over the vector
space of $\bx_i$ that is 
$\nabla_i(\bx) = (\partial \bx/X_{i\,1}, \partial \bx/X_{i\,2}, \ldots)$.

\begin{theorem}
If $\Phi(\bX)$ is continuously differentiable, the dynamical system 
$\dt{\bX} = \Phi(\bX)$ is Lyapunov stable if there is a 
continuously differentiable set of functions $V(\bx_i)$ over
$\bD$ such that for all $i$
\begin{enumerate}
\item $V(\bx_i) \geq 0$ for all $\bx_i \in \bD$.
\item $V(\bx_i) = 0$ if and only if $\bx_i$ is in equilibrium and 
\item $\nabla_i V(\bx_i)\cdot \Phi(\bx_i) < 0$ for all $i$ if $\bX$ is not in equilibrium.
\end{enumerate}
\label{theorem:stability}
\end{theorem}

\begin{proof}
See appendix \ref{sec:smith_extension}.
\end{proof}

A suitable set of $V(\bx_i)$ is given by $V(\bx_i) = 
\sum_{j,k \in \mC_i} X_{ij} (p_k - p_j)_+^2$.  This
is shown in appendix \ref{sec:smith_extension}.

Although this
proof regards a continuous dynamical system, our simulation model,
along with any implementation, would be discrete (see section \ref{sec:realism}).
The \emph{discrete shadowing lemma} shows that, in a wide 
variety of circumstances, the behaviour of a discrete approximation of a continuous 
system can remain arbitrarily close that of the original one
(see Theorem 18.1.3 in \cite{katock1995dynamic}).  Note, however, that the
system proposed may not be such an approximation.

\subsection{The model in practice}
\label{sec:realism}

It is more useful to use split
proportions $s_{ij}=X_{ij}/d_i$ since $d_i$ will, in fact,
vary between days.  This $s_{ij}$ can then be directly interpreted
as the proportion of traffic choosing amongst choice set $\mC_i$ which
should be assigned to slot $j$.

Converting the dynamic system from section \ref{sec:basicmodel} 
from a continuous time to a discrete time (day-to-day) model 
is achieved as follows.
For every choice set $\mC_i$ then a
descent direction is given
by the vector $\dt{\mathbf{s_i}} = (\dt{s}_{i0},\dt{s}_{i1},\dt{s}_{i2},\ldots)$ 
where
$$\dt{s}_{ij} = \sum_{k \in (\mC_i \backslash j)} \left[ s_{ik}(p_k-p_j)_+ -
s_{ij}(p_j-p_k)_+\right].
$$
Moving a delta in this direction will move the $s_{ij}$ closer to 
equilibrium.  The ideal change in
this direction would be:
$$
s_{ij}:= s_{ij} + k_i \dt{s}_{ij},
$$
for some $k_i$.  The $k_i$ must have the following properties:
\begin{enumerate}
\item All $s_{ij}$ must stay in the range $[0,1]$.
\item $k_i \dt{s}_{ij}$ must remain scale invariant with respect
to price (that is, if all prices in the network are multiplied by the same
constant, then $k_i \dt{s}_{ij}$ will not change).
\item Similarly, if the pricing on a single link is multiplied by a constant
factor then $k_i \dt{s}_{ij}$ should not change for switches of slot
in time only. 
\item Neither ``too large'' nor ``too small'' a proportion of traffic must be
moved in every iteration.
\end{enumerate}

This is achieved as follows.  Firstly, define for each $\mC_i$ the
norm $||\mC_i|| = \sqrt{\sum_{j \in \mC_i} \dt{s}_{ij}^2}.$
Note that because of the definition of $\dt{s}_{ij}$ if all prices
in the network which apply to slots in $\mC_i$
are multiplied by a factor $C$ then $||\mC_i||$ is also multiplied by
the same factor.  If the constant $k_i$ is proportional
to $1/||\mC_i||$ then the value of $k_i \dt{s}_{ij}$
is invariant with respect to a multiplication
of the price achieving 2) and 3) above.  Secondly, the equations already
ensure $s_{ij} \geq 0$.  When $k_i\dt{s}_{ij}$ would cause $s_{ij} > 1$ 
then it is reduced accordingly by multiplying $k_i$ by a suitable
constant.  Finally, to achieve 4) then $k_i$ is set initially to
$0.1 / ||\mC_i||$ initially reducing logarithmically to 
$0.001/||\mC_i||$ by the end of
the simulation.  This roughly corresponds to 
$10\%$ of the traffic shifting in the earliest iterations and only
$0.1\%$ by the later iterations.  

In a running system, a system manager
would likely want only small (cautious) steps per day as the system
would run for a long time and attempting large shifts in traffic
between days would be unnecessary.

\section{Modelling framework}
\label{sec:modelling}
The scheme is assessed by evaluation based around real traffic traces.
In order to test how the scheme would perform in a variety of different
scenarios a number of assumptions about pricing are tested.
The traffic traces are analysed to see how they would respond
to the algorithm given hypothetical assumptions about pricing and about
what traffic can be moved in time and what traffic can be moved
in space.

\subsection{Model inputs}
\label{sec:mod_inputs}

Two data sets covering user traffic demand were used.
The first data set (EU) was collected from a large European ISP, and contains
mostly residential traffic.
The second data set (JP) was obtained from non-anonymised traces collected from
within a Japanese academic research network.
Each data set provides the amount of data sent and
received by every network user aggregated over each time window.
In the EU case the traffic is all traffic leaving and entering the
ISP network (over transit and peering links) aggregated for each user 
in hourly periods.
In the JP case the traffic is collected from a single link which all
traffic into and out of that network must traverse and is aggregated
per IP address inside the network in 15 minute periods.  In both cases
cost information and mapping of traffic to external links is
not known.
Summaries for both EU and JP datasets are shown below.

\renewcommand{\tabcolsep}{0.15cm}
\begin{center}
\footnotesize
\begin{tabular}{rccccrr}
& & & & & \multicolumn{2}{c}{Data (GB)} \\
\cline{6-7}
Dataset &  Start date    &  Days    & Window & Users & Out & In \\
\hline
EU & 07/11/11 & 7 & 1 hour & 37,580 & 2,343 & 12,672 \\
JP & 17/03/08 & 3 & 15 min & 12,728 & 849 & 566
\end{tabular}
\end{center}

The EU dataset is an example of a typical eyeball ISP, with inbound
traffic far higher than outbound. Although 
TARDIS is expected to be especially useful for providers like these,
even for networks where such behaviour is not observed traffic
balancing can still reduce costs.
The JP dataset exhibits higher outbound traffic
because demand is driven by remote hosts spanning multiple timezones.
Its inbound traffic, on the other hand,
displays a much stronger diurnal pattern and typically exceeds outbound
traffic during traditional peak hours.
Since TPG and pricing information are unavailable for either dataset, 
it is necessary to make some further assumptions in order to test the TARDIS
procedure.

\subsubsection{Assigning traffic pricing groups and prices}
\label{sec:links}

The JP dataset provides full packet header traces, and thus allows
traffic to be further aggregated by remote host.
By using geolocation on the source IP address, inbound data can be split
into hypothetical ingress TPGs.
Three sets were defined according to the geographic source of data:
\emph{Japan / China}, \emph{United States} and all \emph{other} sources.
This selection reduced the disparity in traffic volume between sets but
also mirrors an operational reality, namely that traffic from regional
and international partners is priced differently.
The resulting geographical bundling is shown in figure \ref{fig:mawi}
and used as a proxy to define traffic pricing groups.

\begin{figure}
\centering{
\includegraphics[width=8cm]{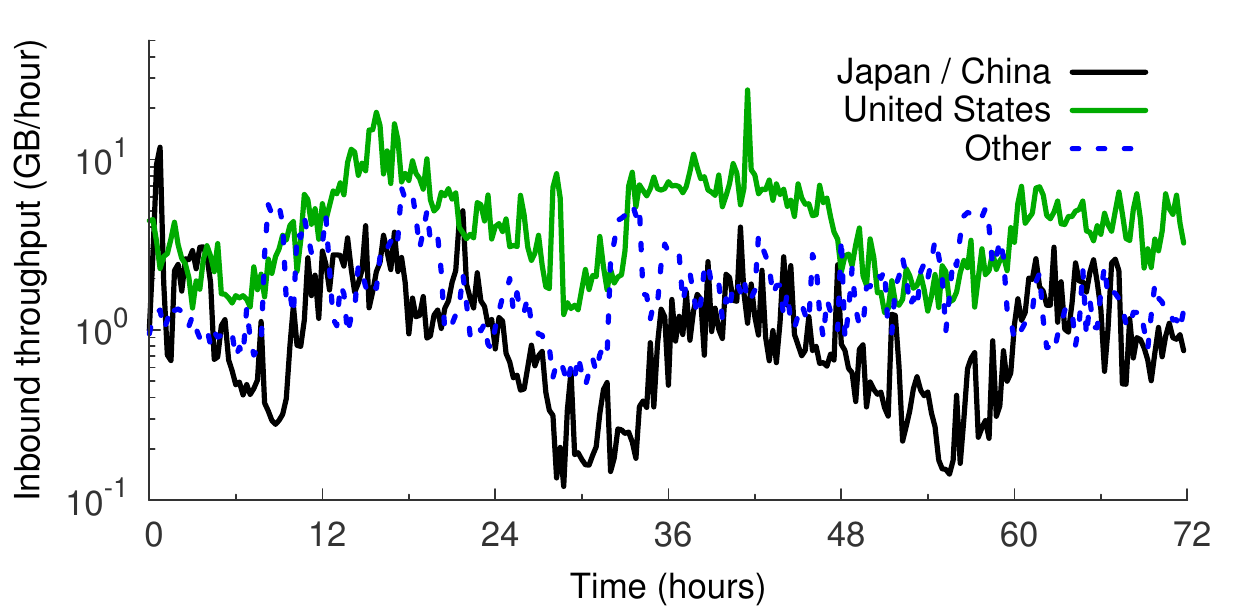}
\caption{Inbound traffic for the JP data set split by geographic
location}
\label{fig:mawi}
}
\end{figure}

For the EU dataset only user traffic aggregates are available and
no destination information was available to split traffic into TPGs.  
It would be a pessimistic assumption for TARDIS to simply split randomly
as this would give each TPG the same temporal profile and reduce the
opportunities for cost saving by space shifting (since each link would
have its peak hour at the same time).  Real traffic traces have different
time behaviour for traffic originating from different destinations
(for example P2P users turning off their client according to their local
diurnal pattern).  Therefore a means of splitting traffic was required
which would induce slight difference in the diurnal cycle.  The following
procedure was used:
\begin{enumerate}
\item Half the traffic was split equally between each TPG.
\item The other half of the traffic was split between each TPG according
to a weighted cosine function with 24 hour periodicity 
$\cos[ (t-p_i)*2 \pi/24]$ where $t$ is the time period for the traffic being split
(in hours) and $p_i$ is the time in hours hour where TPG $i$ gets peak
weight.
\end{enumerate}

This splitting process keeps the total amount of traffic in each hour the same
but ensures that different TPGs have slightly different peaks by choosing
different values for $p_i$.
To evaluate the impact of the difference in the diurnal patterns between TPGs
two different policies are used.
In the policy $T_2$ the
three TPGs have peak hours two hours apart and in $T_4$, four hours 
apart.  Obviously, the more widely separated the peak hours, the
better the system will perform under space shifting.
Figure \ref{fig:europe} shows inbound data for the EU dataset with the
$T_2$ peak shift policy over the three hypothetical TPGs.

\begin{figure}
\centering{
\includegraphics[width=8cm]{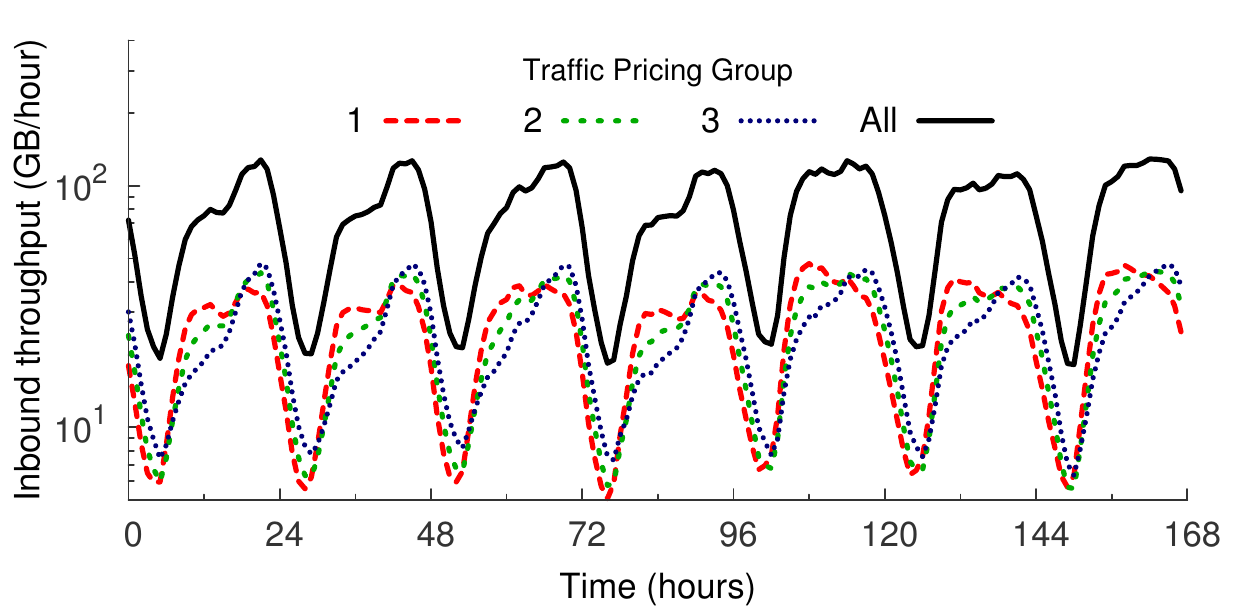}
\caption{Inbound traffic for EU dataset split across three TPGs using
$T_2$ weighting policy.}
\label{fig:europe}
}
\end{figure}

Having constructed hypothetical traffic pricing groups, it is now
necessary to assign pricing schemes to each group.   Hypothetical
schemes were assumed to test a variety of scenarios (no real pricing
data was available for the networks being tested).
It will be assumed that traffic will be charged at the 95th percentile
level on all TPGs, for the following pricing levels:
\begin{itemize}
\item \textbf{Equal Prices}, $P_{E}$, TPGs priced in the ratio 1:1:1.
\item \textbf{Low Variation}, $P_{L}$, TPGs priced in the ratio 4:2:1.
\item \textbf{High Variation}, $P_{H}$, TPGs priced in the ratio 10:3:1.
\end{itemize}
The high variation corresponds to ``We set the unit traffic costs to
be between \$1 and \$10 per Mbps, which corresponds to publicly
available data on the current prices for Internet transit" 
\cite{lakhina_cost_2012}.

\subsection{Modelling user choice}
\label{sec:choice}

The model discussed above can be used to determine how much demand a
user places on traffic pricing groups for each time window.
The next stage is to consider which other slots traffic may be shifted
to.
In practice, this will depend on technical feasibility, user willingness
and the availability of data in another slot.

A question arises as to the correlation between traffic which can
be shifted in time and traffic which can be shifted in space.
If it were unlikely that time-shifting traffic could also space shift
(because the two types were fundamentally different) then the 
two would be anti-correlated and this would mean
that more traffic overall could shift than if there were no 
correlation.
Conversely, it may be that the opposite is true and that traffic which can
shift in time is more likely to be able to shift in space.  If this is
the case then less traffic will be able to shift at all than if there
were no correlation.  For the model we use we take the independence
assumption as being a mid-way assumption.

An important problem in using the model in a practical context
is generating the choice sets. To approach this problem
systematically choice sets are generated separately
for demand which can shift in time, demand which can
shift in space and demand which can shift in both.
For example, for space shifting, 
there must be a choice set for each time window
containing slots for all TPGs.
For time shifting within a given TPG, 
there must be a choice set for each time window 
which contains slots representing that TPG at that
time window, the next time window, the time window
after that, up to the maximum allowable delay shift.  Although
the number of choice sets is large, it is manageable
as it is proportional to the product of the number of time windows and TPGs.

The user choice model selected can be condensed into
the small set of parameters shown below.
\begin{center}
\footnotesize
\begin{tabular}{r|p{6.5cm}}
$N_T$ & Proportion of users eligible for time shifting. \\
$N_S$ & Proportion of users eligible for space shifting. \\
$\mu_T$ & Mean proportion of data shifted in time by eligible users.\\
$\mu_S$ & Mean proportion of data shifted in space by eligible users.
\end{tabular}
\end{center}
Given these parameters, each user selects their time and
space shifting characteristics as follows:
\begin{enumerate}
\item The ability to shift traffic in time or space is determined by
comparing randomly generated numbers in (0,1) to thresholds $N_T$ and
$N_S$ respectively.
\item For time shifting users, the proportion of time shifted
data is set individually for
each user to $P_T=R^{(1-\mu_T)/\mu_T}$ where $R$ is a random number in
(0,1). This produces a population of users who swap up to 100\% of their
traffic in time, with a mean proportion of $\mu_T$.
\item Similarly, for space shifting users, the proportion of
space shifted data by each user
is set to $P_S = R^{(1-\mu_S)/\mu_S}$.
\end{enumerate}

An assumption needs to be made about the maximum time delay allowed
(here 18 hours was chosen) and which links are available for traffic
to choose (here, it was chosen that traffic which could shift in space
could always choose all three links).
The combination of both available time windows and TPGs produces a
choice set of available slots.
For simplicity, traffic is treated identically within each slot.
Given the available slots, the proportions $p_{ij}$
assigned by the dynamic model
in section \ref{sec:dynamics} are used to choose which slot to assign
the traffic to.
An amount of traffic $f$ can be assigned among slots in choice set
$\mC_i$ in two possible ways:
\begin{itemize}
\item \textbf{All-or-nothing assignment}. All of $f$ is assigned to a
randomly chosen $j \in \mC_i$ with probability $p_{ij}$.
\item \textbf{Proportional assignment}. A proportion $p_{ij}$ of $f$ is
assigned to each slot $j \in \mC_i$.
\end{itemize}

\subsection{Data analysis details}

After the previous assumptions are made, the processing of data to assess
the TARDIS algorithms is as follows:
\begin{enumerate}
\item \textbf{Initialisation}. A preliminary run through of all 
$N$ days of
traffic is used to initialise the traffic splits $s_{ij}$ and to
generate initial traffic profiles.
\item \textbf{Pricing}. Based on the traffic from the previous day, 
costs for each slot are calculated using the Shapley gradient procedure
presented in section \ref{sec:pricing}.
\item \textbf{Traffic Shifting}. The slot costs are used to
update $s_{ij}$ for all slots in $\mC_i$ as explained in section 
\ref{sec:dynamics}. Traffic assigned to each slot
is updated accordingly.
\item \textbf{Iteration}. The next day is processed by returning to step
2. If day $N$ has been reached, processing continues by wrapping
around to the first day, the second day and so on.
\end{enumerate}

\begin{figure}[ht!]
\begin{center}
\includegraphics[width=8.5cm]{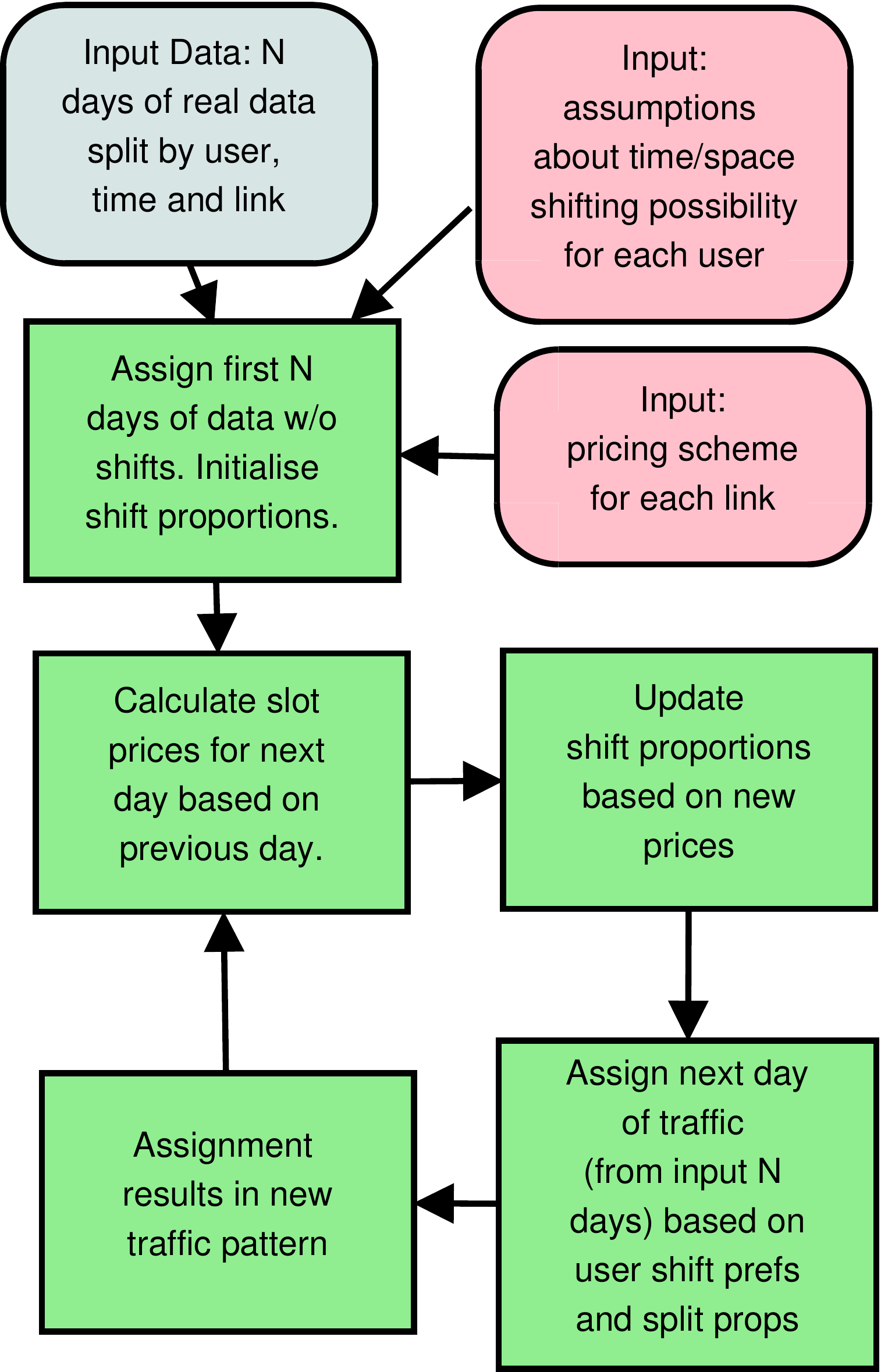}
\caption{Diagram showing the trace driven analysis of TARDIS.}
\label{fig:model}
\end{center}
\end{figure}

Figure \ref{fig:model} shows the analysis in diagramatic form.
The known input is the $N$ days of traffic (from EU or from JP).
This will include the assumptions $T_0, T_2$ or $T_4$ from
section \ref{sec:mod_inputs} for the EU data and the geographic
split for the JP data.  Two assumed inputs are the time/space
shifting possibilites (as discussed in section \ref{sec:choice})
and the pricing schemes.  All pricing schemes tested here
are 95th percentile (as that is the main focus of the Shapley gradient
method) and the ratios between TPG are either $P_H$, $P_L$ or $P_E$
with $P_H$ the highest difference between links being the base case.

The analysis begins by modelling the traffic as if it could not
choose slots for a warm up period.  This initialises the shift
proportions $s_{ij}$ for all choice sets $\mC_i$.  
These are used to generate slot prices for the next
day according to the Shapley gradient $G_j$ measured on the
previous day's traffic.  The slot prices are used to update the shift
proportions and these updated shift proportions are used to generate
the next day of traffic in combination with the next day of traffic data
from the EU or JP dataset.

Finally it is worth briefly mentioning execution time.  The algorithm
given is easily lightweight enough 
to run real time in practice for the number of users discussed
here.  The graphs in the next section each contain 25 graphs of 500 
days and were run on only modest computing hardware.  Each 500 day run
took less than 1 day of real time to complete for 10,000 users.  
An implemented system would obviously have to perform only 1 day per day.
It should also be noted that if the number of users is large then splitting 
rates can be calculated using only a sample of data.
In fact the runs for the EU data were
performed on a sample of 10,000 users -- runs on the full data set produced
extremely similar results.  The runs for the JP data were
performed on the full data set.

\section{Analysis of user data}
\label{sec:results}
Results are presented in a common format.  The simulation is run for
50 days as a ``warm up" and then 500 simulated days.   The effects of 
this recycling days is discussed in section 
\ref{sec:assumptions}.
The split ratio
is frozen for the final 50 days and the price paid over the first
50 days is compared with the price paid over the final 50 days. 
The graphs is
repeated for different proportions of traffic being allowed to swap in
time and in space.  It is uncertain what proportions of traffic could
in potential be engineered in this way so values from 0 to 60\%
are tried for space swapping and values from 0 to 20\% for time swapping.
The results are presented as the proportional reduction in price paid by
the ISP and are compared to the theoretical maximum efficiency.   The
proportional reduction is given by 
$(P_i-P_f)/P_i$ where $P_i$ is the mean daily price
for the initial 50 day period and $P_f$ is the mean daily price for the final
50 day period.  This ratio represents the reduction in price paid, e.g.
0.4 is a 40\% saving.
The theoretical maximum efficiency is simply the proportion
of traffic which can swap either in time or in space.  The data will
not necessarily produce the same result on repeated runs as there are 
random elements to which users are assigned to swap and different users
have different data profiles.  Repeated tests with the same input are
used to estimate the standard deviation on the obtained results and
these are included on all plots.  The exact procedure is described in 
section \ref{sec:assumptions}.  Each box in the plots, therefore has
the form $x \pm y (z)$ where $x$ is the calculated reduction in the price,
$y$ is twice the estimated standard deviation (if the variable were normally
distributed this would be a 95\% confidence interval 
and $z$ is the reduction which 
would be achieved if all shifted traffic had zero cost.  So, for example,
$0.1 \pm 0.02  (0.14)$ should be interpreted as a cost reduction of 0.1
(10\%) on the base costs with the true figure likely to be between 
0.08 and 0.12 (the range being two standard deviations each side of
the mean) and the maximum possible reduction (if all shifting traffic
simply disappeared) being 0.14.

In all the cases tested here it is assumed that all three TPG are available
for swapping.  While this is an optimistic assumption, in many real situations
more than three TPG would be available to an ISP.
The high variation
pricing scheme $P_H$ (i.e. a price ratio of 10:3:1) using
95th percentile pricing is the default for
the base scenario. In addition, for the 
EU data set the base  model includes the 
peak hour weighting scheme $T_2$ (see section \ref{sec:links}). 
The base model uses the proportional assignment 
choice model (see section \ref{sec:choice}) but this was found to make
no difference to the outcome (see section \ref{sec:assumptions}).

The proportion of time shifters is a combination of the proportion
of users whose traffic shifts in time and the proportion of traffic
shifted in time $N_T\mu_T$.  Similarly for space shifting the 
proportion of traffic shifted in space is $N_S\mu_S$.  To reduce
the number of modelling parameters varied then $N_T = N_S = 1$
and $\mu_T$ and $\mu_S$ are varied.  However, the results were
found to be no different if $\mu_S = \mu_T = 1$ and $N_T$ and $N_S$
are varied (see section \ref{sec:assumptions}).

\subsection{Base case results}

\begin{figure}
\centering{
\includegraphics[width=8.5cm]{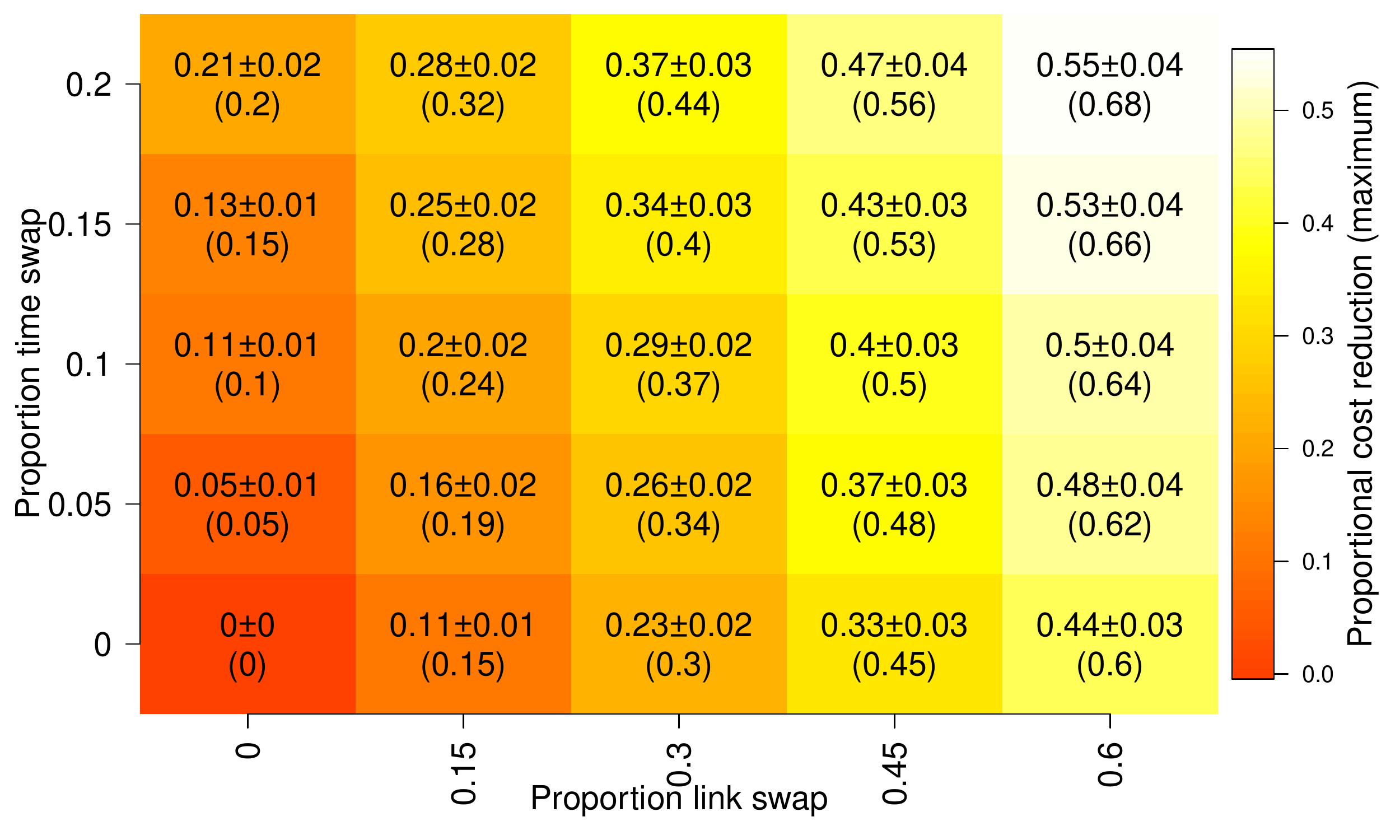}
\includegraphics[width=8.5cm]{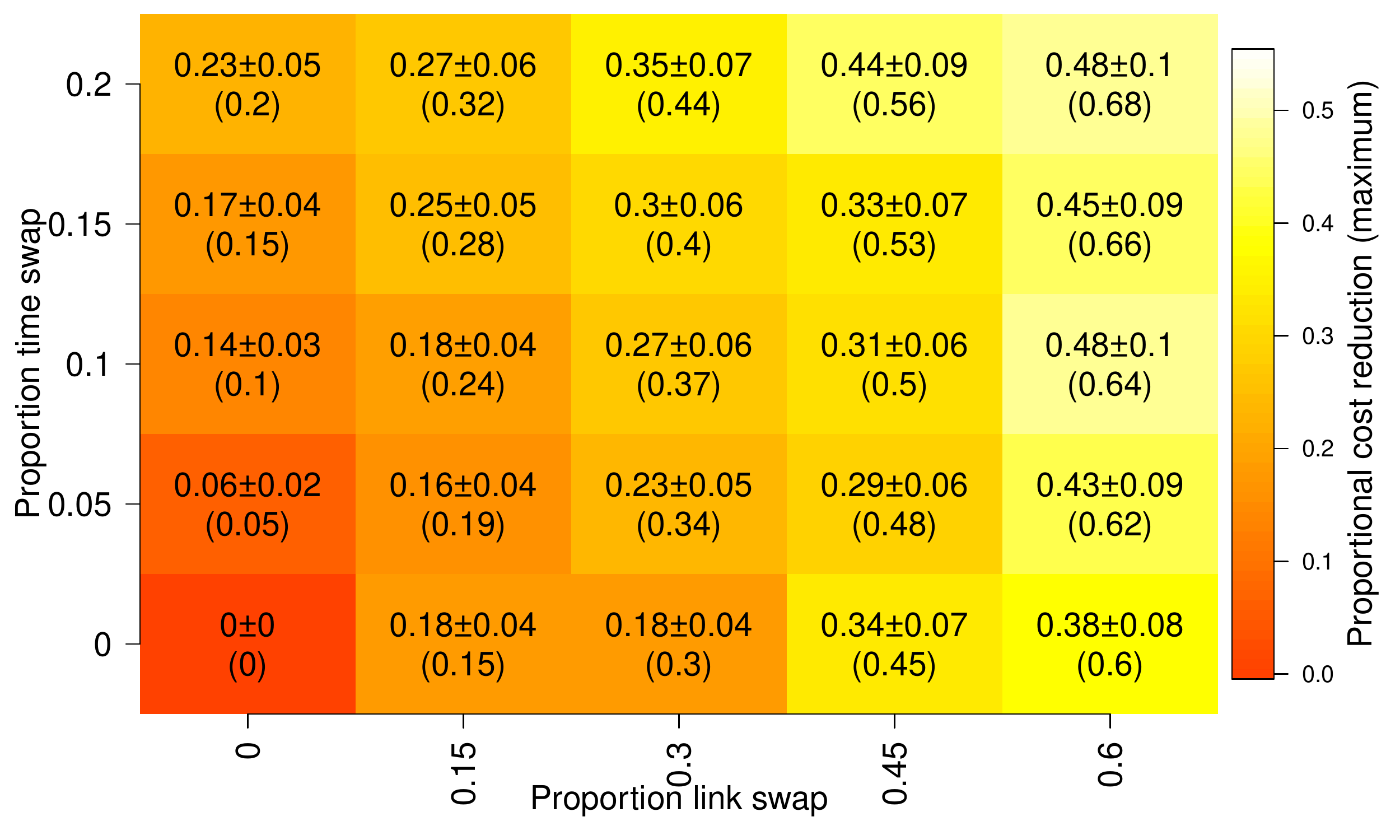}
\caption{Base case for EU data (top) and JP data (bottom).}
\label{fig:base}
}
\end{figure}

This section presents the results on the base case analysis.
Figure \ref{fig:base} (top) shows the base case for the EU data.  As can
be seen, the savings are substantial with, in the highest case, 55\% of
the transit price being reduced.  In many cases the saving is close to
the optimal pricing that could occur.  Note that in some cases it is slightly
over due to statistical variances in the results (for example, if the traffic
randomly selected to be allowed to shift happened to be on higher priced
links in that run).
It can be seen that, in particular, time shifting is extremely good at
saving cost with almost all of the potential benefit being realised.

Figure \ref{fig:base} (bottom) shows the same results for the JP data.
In this data there is more statistical variance in the results, 
likely because the data set has more ``imbalance" in the traffic 
distribution with a small number of users contributing a large amount
of data.  Naturally, if those users move their data then this contributes
more to the solution.  
The results show broadly the same pattern as the
EU data but the link swapping is slightly less effective.

\subsection{Variant results}

A number of variants on the base scenario were tried and are reported here.
The most obvious variant is to try the lower difference in pricing.
In this scenario the link prices are split in the ratio 4:2:1 not
10:3:1.  The results are presented in figure \ref{fig:low}.  
The effectiveness of time swapping is not reduced significantly as 
might be expected.
However, the effectiveness of link swapping remains relatively high,
especially in scenarios with low levels of time swapping.  Therefore
it seems link swapping can still be extremely effective for saving cost
even when cost differences are low.  In fact there was still some leeway
Nonetheless, in this scenario, in the cast with most swappers, 
almost 40\% of the price is saved in both data sets (the
theoretical maximum being 68\%). In fact tests on the equal price
policy $P_E$ show that by taking advantage of the peak time differences
gains can be made even when there is no difference between the
rates charged.  For example on the JP data set a cost saving of 10\%
was made in a scenario with 20\% space shifting and no time shifting.

\begin{figure}
\centering{
\includegraphics[width=8.5cm]{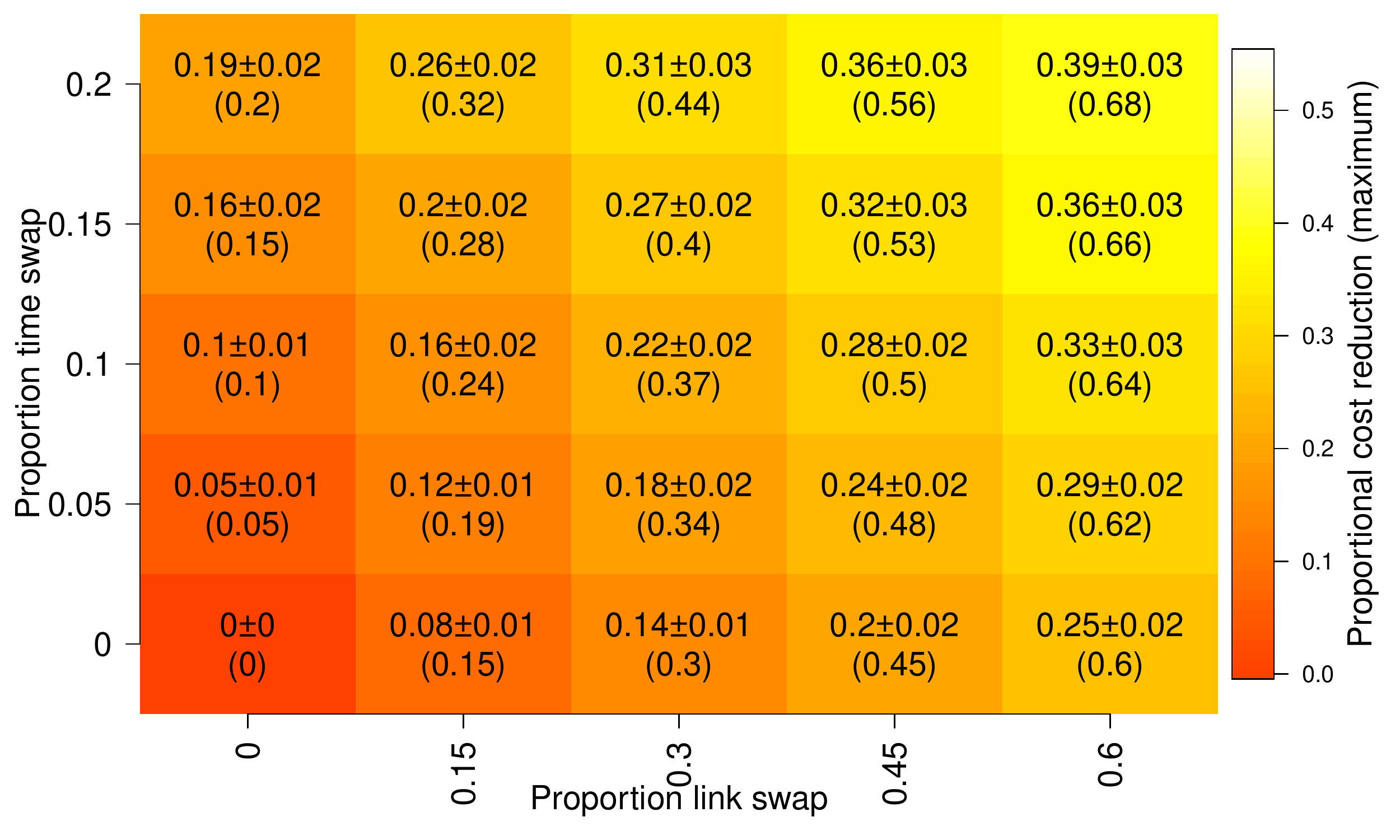}
\includegraphics[width=8.5cm]{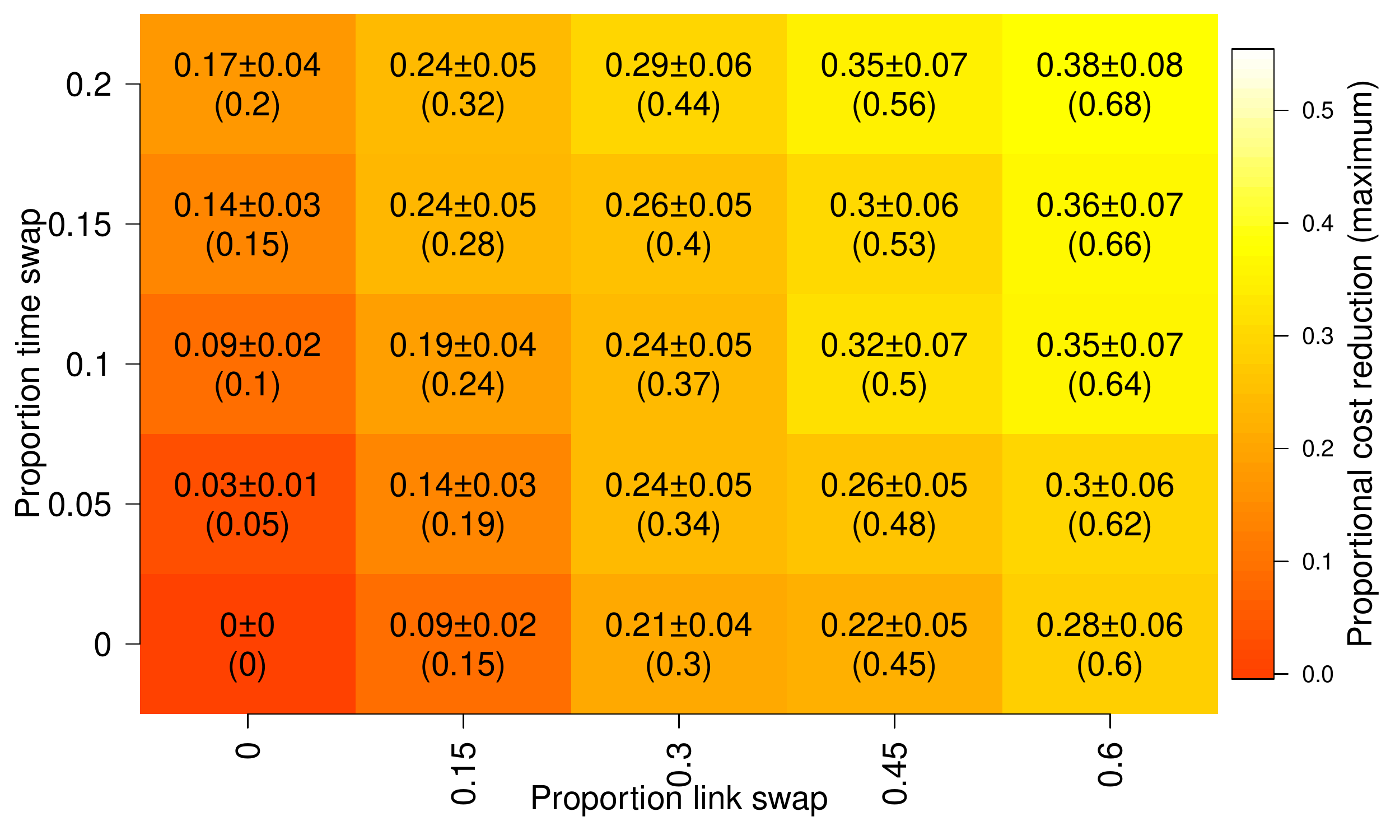}
\caption{Low pricing variant case for EU data (top) and JP data (bottom).}
\label{fig:low}
}
\end{figure}

An investigation was made into the effect of 
a wider split in the peak hours across pricing groups, the $T_4$ 
scenario described in section \ref{sec:links}.  As would be expected,
this produces some advantage for the TARDIS system but in fact
the advantage is slight and the results are indistinguishable from the
$T_2$ results when error bars are accounted for.  The graph for
the EU data is shown in figure \ref{fig:4hours}.  When compared
with figure \ref{fig:base} (top) it is apparent that little difference
to the price has been made.

\begin{figure}
\centering{
\includegraphics[width=8.5cm]{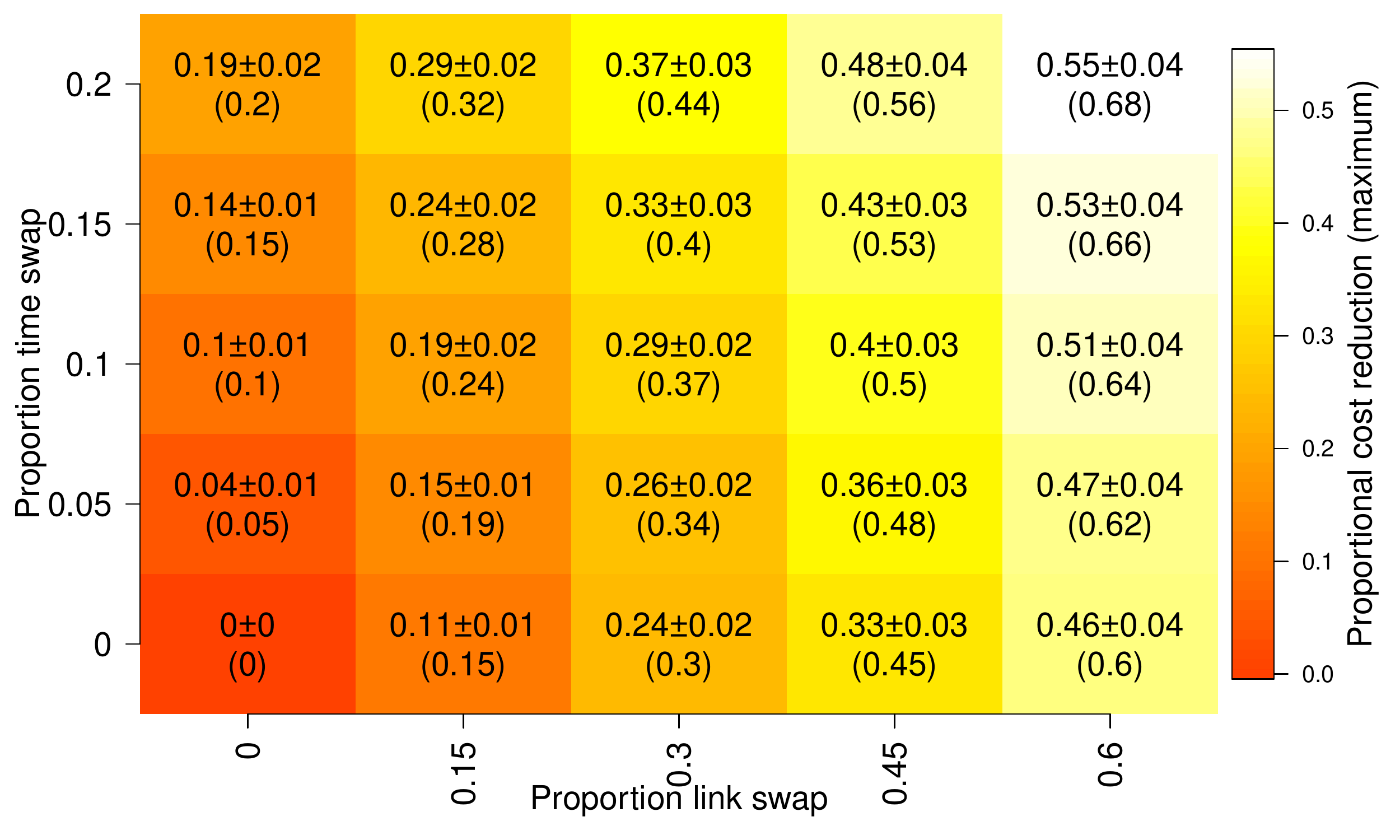}
\caption{The $T_4$ variant in traffic splits, EU data}
\label{fig:4hours}
}
\end{figure}

\subsection{Checking data analysis assumptions}
\label{sec:assumptions}

A number of assumptions on the data analysis are checked for robustness.
The first, and most important, is the repeatability of runs.  As mentioned,
the stochastic element to assignment of which traffic can be swapped
means that not every run gets the same results even with identical input.
It is therefore important to calculate the repeatability of the results.
This is assessed by a run with 60\% link swappers and 20\% time swappers
repeated ten times each
for the JP and EU data.  The Shapiro-Wilk normality test
showed that the results were not normally distributed and no
simple distribution could be found hence no simple way
of calculating confidence intervals was available.  Instead the
coefficient of variation is calculated (the ratio of standard deviation to mean).
This is used to calculate estimated standard deviations for the results.

Runs were made with prices on all links equal.  Benefits from swapping 
in space were still present as different links had different peak
times.  The exception was in the EU data set if policy $T_0$ was
used to split traffic between links.  In the case of exactly equal traffic
on all links and equal prices on all links then no benefit was discovered
from link swapping, as might be expected.  With equal prices, 20\% of link swappers and no time
swappers produced a 5\% saving in prices in the EU data (with the $T_2$ policy)
and 10\% in the JP data.

Runs were made to test assumptions about which traffic was chosen to shift.
No significant difference was found in the results when $x\%$ of users
were chosen to shift all their traffic or when all users were chosen to
shift $x\%$ of their traffic (chosen at random).

Runs were made to test assumptions about the limited number of days of data
available.  In specific it might be worried that the very good performance
on the data was due to the same data being recycled again and again.
To test this artificial extra days were generated from the real days of
data by the following procedure.  
\begin{enumerate}
\item A total traffic level for the new day is chosen from a normal
distribution with the same mean and variance as the real traffic's 
daily traffic level.
\item Each user picks one day at random and uses their traffic profile
for that day.
\item The traffic for each user is multiplied by a constant chosen to
give the required traffic level from the first step.
\end{enumerate}
This procedure will generate extra days of traffic which have the same mean
traffic level and same variance in traffic level between days as the original
and will also have the same split of traffic between users as the original.
However, every day of traffic will be different.  By using this procedure
then it is certain that cost savings are not due to the traffic being
``too predictable".  Conversely, however, the traffic levels for the
assessment days can increase or decrease (about a constant mean) so 
a perceived cost saving or increase could be a result of a temporary
increase in traffic.  Figure \ref{fig:mawi_remap} shows these results
for the JP data set.  This should be compared with figure \ref{fig:base}
(bottom).  The errors are larger on this graph because of the
fluctuation on the data.  This can be seen especially in the base
case with zero time and space shifting
where the cost has got worse even though no users can shift.  This
is simply a random fluctuation upwards in traffic.  The two standard
deviation bound is larger in this graph because it includes elements
for the repeatability of the assignment procedure but also variations
in the traffic levels.  There appears to have been a decline
in the benefit in some scenarios this is within the two standard
deviation bounds so it is hard to tell whether it is a real
variation or a result of random traffic growth.  For example
$0.48 \pm 0.1$ in the base case has become $0.41 \pm 0.14$ in
the case with varying traffic.  Most falls are of around this
level so if the variation in the traffic is causing worse performance
it seems that it is not greatly worse.

\begin{figure}[ht!]
\centering{
\includegraphics[width=8.5cm]{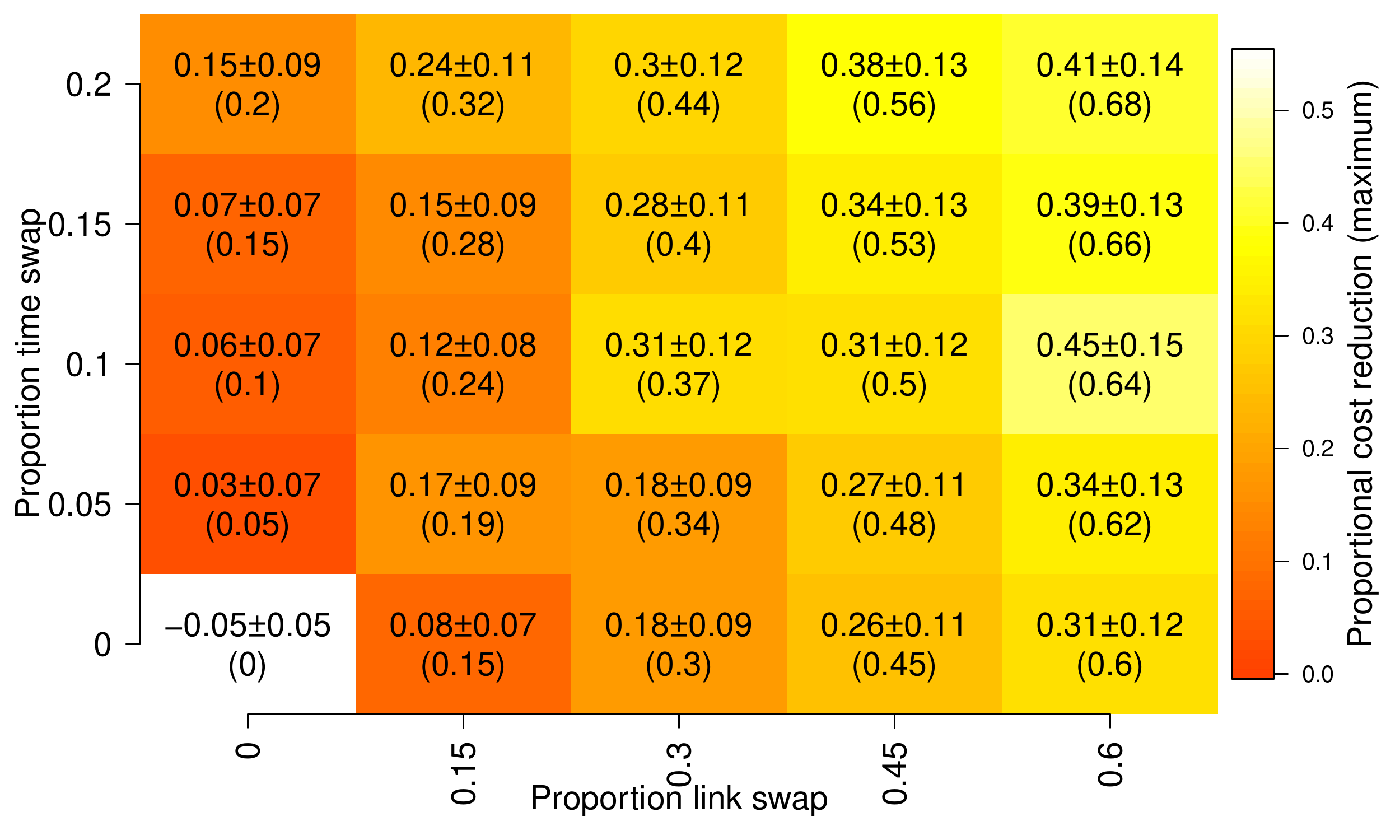}
\caption{The JP results with randomised ``extra" days of traffic.}
\label{fig:mawi_remap}
}
\end{figure}

Finally, the outcomes of using either the all-or-nothing or 
proportional allocation strategies (see section \ref{sec:choice}) 
were also tested. No significant differences were found between the 
two.

\subsection{Discussion and criticism of results}

It is hard to get realistic results to assess the TARDIS system
for several reasons.  Firstly, 
it is hard to get data sets for several days of data
which maps ingress and egress traffic 
to individual users (ideally in a non-anonymised way the lack 
of anonymisation is what enabled separation of the JP traffic 
by destination).  Publicly available traces known to the researchers
were unsuitable for one or more of these reasons and hence non
public data had to be used.
Secondly, ISP pricing plans are commercially sensitive and not usually
publicly disclosed.
Findally, knowledge of the likely traffic which can be shifted is hard
to get.  The percentage shifting in space
can only be estimated from work such as \cite{ager2011,frank2012cate,poese2010improving}, estimates seem
to be relatively high (\cite{frank2012cate} gives 40\% of traffic available
in three or more locations).
Knowledge of which traffic can be time shifted is even less
available although some insights can be gained from work such as
\cite{wong2011time,wong2011time2}.

In the light of these problems, the results in this section are
an attempt to get the best possible assessment of the system without
introducing too many parameters which cannot be estimated.  They 
should be taken as an investigation of how well the system is likely
to work under a range of conditions and a robust attempt to locate
potential areas where the assumptions would cause problems for
such a system in practice.  The large scale of the traces meant
that each investigation took considerable computing time 
(This is because of the necessity of simulating the behaviour of tens
of thousands of users for every day simulated, a real TARDIS system
would only have to calculate splitting rates which is a much easier
task.)

This said, the results are remarkably successful in all the 
system variants investigated here.  In the majority of cases
the TARDIS system extracted a high proportion of the maximum
possible benefit available.  Even when the assumptions were
relatively conservative (a low variation in pricing between
transit links, a small percentage of traffic able to swap in
space or in time) the benefit in terms of cost
saving could still be relatively large.

With more detailed knowledge of ISP pricing and internal
network structures
the scheme could be tested on its ability to reduce transit costs
while retaining the constraints of limited traffic capacities on internal
links or while avoiding causing excess congestion in downstream
systems.

\section{Conclusions}
\label{sec:conclusions}
This paper introduced TARDIS 
(Traffic Assignment and Retiming Dynamics with Inherent Stability)
an algorithm for determining how to reassign traffic in time and
space to reduce ISP transit costs.  A method was given to assign a
cost to a given link at a given time of day according to any of
a number of widely used pricing schemes.  This time of day cost
was used as an input to a reassignment scheme for traffic.  Modelling 
the scheme as a dynamical system it was shown that a continuous 
approximation to the scheme was provably stable and assigned
traffic to an equilibrium situation.

The scheme was tested in a realistic context by analysis of real life
data sets.  This analysis tested several different assumptions about
pricing levels and about proportions of traffic space and time shifting.  In the
majority of cases a large financial saving is possible.  Time shifting
appears to create a considerable saving in most situations.  Space
shifting creates a saving in all situations except those where
all links are equally priced and traffic is split equally across
all links.

\emph{This research has received funding from the Seventh Framework 
Programme (FP7/2007-2013) of the European Union, through the FUSION
project (grant agreement 318205).}

\bibliographystyle{abbrv}
\bibliography{sigmetrics_tardis_2014}

\appendix

\section{The Shapley gradient price}
\label{sec:shap_indep}
In section \ref{sec:shapleygradient} the Shapley gradient $G_j$ was 
introduced as the cost gradient of the Shapley value when a fictitious
user $N+1$ adds an amount of traffic $du$ to slot $j$.
Define the per user Shapley gradient for user $i$ in slot $j$ as
$\phi'_{ij}= (\phi_i(d_{ij}) - \phi_i) / du$ where
$\phi_i(d_{ij})$ is the Shapley value for user $i$ with
extra traffic $du$ in slot $j$.  It was stated that for all
schemes considered in this paper then the mean over
$i$ of $\phi'_{ij}$ is $G_j$.  In the simplest cases,
$\phi'_{ij} = G_j$ for all $i$.

Formally, from section \ref{sec:shapleygradient}
\begin{align*}
\nonumber G_j  &=  
\frac{1}{du\,(N+1)!} \\ 
& \sum_{\pi \in \mS_{\mN'}} [v(\mS(\pi,N+1)) - 
v(\mS(\pi,N+1) \backslash N+1) ],
\end{align*}
where 
$\mS_{\mN'}$ is the set of arrangements of the
users $\mN$ plus the fictitious $(N+1)$th user and 
$\mS(\pi,N+1)$ is the set of all users arriving
not later than user $N+1$ in the permutation $\pi$.

For linear pricing this is trivial to show.  If the slot
$j$ is charged at rate $r_j$ then $\phi_i (d_{ij}) - \phi_i =
r_j du$ since the extra traffic $du$ costs $r_j du$.  Hence,
$\phi'_{ij} = r_j$ and is not dependent on $i$.  
It is already shown in section \ref{sec:linear_price}
that $G_j = r_j$.

For the pricing of a link of fixed capacity as described in 
section \ref{sec:fixed_price} then the case is very similar.
The price depends only on the total flow on the link.  So
\begin{align*}
\phi_i (d_{ij}) - \phi_i & = \frac{f_j  + du -\alpha m}{m - (f_j + du)} -
\frac{f_j  -\alpha m}{m - f_j} \\
& = \frac{(1-\alpha)m \, du}{(m - f_j) (m - (f_j+du))}\\
& = \frac{(1-\alpha)m \, du}{(m - f_j)^2},
\end{align*}
where the final equality is because $du$ is infinitesimal with
respect to $f_j$.
Hence $\phi'_{ij}=G_j = \frac{(1-\alpha)m}{(m - f_j)^2}$ as
calculated in section \ref{sec:fixed_price}.
Any scheme where the price does not depend on the user could
be analysed in a similar manner.

For the 95th percentile pricing then by the arguments from
section \ref{sec:shap95}
$$
\phi'_{ij} = (\phi_i(d_{ij}) - \phi_i)/du = 
 \frac{1}{N!} \sum_{\pi \in \mS_{\mN}} [ F_{ij} ],$$
where 
\begin{align*}
F_{ij} & = [v(\mS(\pi,i)) + I (j \in \mT^{(95)}(\mS(\pi,i)) r_j\, du \\
& \quad - v(\mS(\pi
,i) \backslash i)]/du  - [v(\mS(\pi,i))  - v(\mS(\pi,i) \backslash i)]/du \\
& = I (j \in \mT^{(95)}(\mS(\pi,i))  r_j.
\end{align*}
The mean of $\phi'_{ij}$ over all users $i$, $\overline{\phi'_{ij}}$
is given by:
$$
\overline{\phi'_{ij}}= \sum_{i=1}^N \frac{\phi'_{ij}}{N} = 
\frac{1}{N N!} \sum_{i=1}^N \sum_{\pi \in \mS_{\mN}} 
I (j \in \mT^{(95)}(\mS(\pi,i))  r_j.
$$
Now notice that 
$\sum_{\pi \in \mS_{\mN}}$ covers all arrangements of $N$ items
and 
$I (j \in \mT^{(95)}(\mS(\pi,i))$ covers all occasions where an
addition of traffic $du$ following user $i$ falls in the 95th
percentile set for $\mS(\pi,i)$.  As this is summed over all
values of $i$ this is exactly the same as considering those
occasions where the $N+1$th user (who adds traffic $du$
to slot $j$) follows any other user over any
arrangement of users in $\mN'$ (the set of users 1 to $N$ plus
the fictitious user $N+1$).  This includes every possible
arrangement of users in $\mN'$ except those with user $N+1$
first.  Therefore
\begin{align*}
\sum_{i=1}^N \frac{\phi'_{ij}}{N} & = \frac{r_j}{N N!}
\sum_{\pi \in \mS_{\mN'}} [v(\mS(\pi,N+1))  \\
& \quad 
- v(\mS(\pi,N+1) \backslash N+1) ] \\
 & = \frac{N+1}{N} \frac{r_j}{(N+1)!} \sum_{\pi \in \mS_{\mN'}} [v(\mS(\pi,N+1)) \\
 & \quad
- v(\mS(\pi,N+1) \backslash N+1) ] \\
 & = \frac{N+1}{N} G_j.
\end{align*}
As $N$ becomes large $(N+1)/N \rightarrow 1$ and hence
the mean over $i$ of $\phi'_{ij}$ becomes closer to $R_j$.
More precisely
 $\overline{\phi'_{ij}} = R_j + O(1/N)$
as stated in section \ref{sec:shap95}.

\section{A stability proof for multiple choice sets}
\label{sec:smith_extension}
Theorem \ref{theorem:stability} is identical to that in the appendix
of \cite{smith1984stability} with the exception of condition (3)
which in that reference is given for a single vector $\bx$ as
$\nabla V(\bx)\cdot \Phi(\bx) < 0$, in the theorem here is given
over several vectors as $\nabla_i V(\bx_i)\cdot \Phi(\bx_i) < 0$
where $\nabla_i$ is the gradient over the vector space of $\mC_i$, 
that is 
$$
\nabla_i(\bx) = (\partial \bx/X_{i\,1}, \partial \bx/X_{i\,2}, \ldots).
$$

Recall that $\bD$ is the set of $X_{ij}$ which are \emph{demand feasible}
(that is $\sum_j X_{ij} = d_i$ for all $i$).  Intuitively, this says
that the demand which must be assigned to choice set $\mC_i$ is
equal to the sum of the demand which is actually assigned.  It is
important to note that the demand feasible set is decomposable by choice
set in the sense that the matrix $X$ is demand feasible if and only
if each of its rows $\bx_i$ is.  Therefore, 
let $\bx_i$ be \emph{demand feasible} if 
$\sum_j X_{ij} = d_i$ and (slightly abusing the notation) say
that $\bx_i \in \bD$.
For points on the trajectory of the dynamical system then
$\dt{\bX} = \Phi(\bX)$ and therefore, for all $i$ then
$\dt{\bx_i} = \Phi(\bx_i)$.  

Having established that a demand $\bX \in \bD$ must be composed
of vectors $\bx_i$ all of which themselves are demand feasible 
($\bx_i \in \bD$) then the proof now follows that in the appendix 
of \cite{smith1984stability} for
each $\bx_i$ component separately.  This shows that each choice set
individually converges to its equilibrium condition and
hence the entire system converges.  

In brief the proof in \cite{smith1984stability} follows
an epsilon delta style argument.  Firstly a set $\bD_\epsilon$ is
defined as the set of demand feasible points $\bx \in \bD$
with $V(\bx) > \epsilon$ for
some $\epsilon > 0$.  
It is then shown that if the system starts at a position $\bx \in 
\bD_\epsilon$ and the
dynamical system evolves following $\dt{\bx} = \Phi(\bx)$ and
stays within $\bD_\epsilon$ until some time $T$ then there
is some fixed $\delta > 0$ for which $\nabla V(\bx)\cdot\Phi(\bx) < -\delta$.
Let $\bx(t)$ represent the system state at time $t$.  The function
$V(\bx)$ must shrink at least at rate $\delta$ and hence
$V(\bx(T)) - V(\bx(0)) < - \delta T$.  Therefore $V(\bx(t_0)) < \epsilon$
for some finite $t_0 \leq V(\bx(0))/\delta$.  In other words 
the points leave $\bD_\epsilon$ within some finite time whatever the
size of $\epsilon$.
Since this argument applies for all $\epsilon$ then as 
$\epsilon \rightarrow 0$ the set $\bD_\epsilon$ covers all of $\bD$ except
for the equilibrium positions where $V(\bX) = 0$.

It is now necessary to show that the candidate 
$V(\bx_i)$ meets the three conditions of Theorem \ref{theorem:stability}.
Recall from section \ref{sec:basicmodel} that the candidate functions are
$$
V(\bx_i)=\sum_{j,k \in \mC_i} X_{ij} (p_k - p_j)_+^2.
$$

The first two conditions are trivially shown.  
Firstly, $V(\bx_i) \geq 0$ since 
both
$X_{ij} \geq 0$ and $(p_j - p_k)_+^2 \geq 0$.  Secondly $V(\bx_i) = 0$
only when $\mC_i$ is at equilibrium.  This follows from the
definition of equilibrium.  If some term is non zero then $X_{ij} > 0$
and $p_j > p_k$ for some $j,k \in \mC_i$ then this 
implies there is a flow $X_{ij}$ 
in choice set $\mC_i$ which has a price $p_j$ greater than some $p_k$
also in $\mC_i$.  This is counter to the definition of equilibrium in 
definition \ref{defn:eqbm}.

The final condition (3) of the theorem is the most difficult and again
the proof follows \cite{smith1984stability} but generalised to several
choice sets.

Begin by differentiating $V(\bx_i)$ by parts.
Define $\Delta_{ijk}$ as a vector with 
$-1$ in position $j$ and $+1$ in position $k$ if $j,k \in \mC_i$ and
the zero vector if either $j \notin \mC_i$ or $k \notin \mC_i$.  
Define $\bp$ as
the price vector $\bp(\bX) = (p_1(\bX), p_2(\bX), \ldots)$.
In this form, then $V(\bx_i) = \sum_{j,k \in \mC_i} X_{ij}(p_j - p_k)_+^2$
can be written more compactly
as $V(\bx_i) = \sum_{j,k} - X_{i\,j}\, (\bp(\bX)\cdot\Delta_{ijk})_+^2$
(where the sum is over all flows in $\bX$ and $\cdot$ is the
inner product).  Considering $\Phi(\bx_i) = (\dt{X_{i\,1}}, 
\dt{X_{i\,2}}, \ldots)$ the elements 
$$
\dt{X_{ij}} = \sum_{k \in \mC_i} \left[ X_{ik}(p_k-p_j)_+ -
X_{ij}(p_j-p_k)_+\right]$$
can be written as
as 
$$
\dt{X_{ij}} = \sum_{k} X_{ij} (\bp(\bX)\cdot\Delta_{ijk})_+
- X_{ik} (\bp(\bX)\cdot\Delta_{ijk})_+.
$$

Performing the differentiation $\nabla_i V(\bx_i) $ by parts gives
\begin{align*}
\nabla_i V(\bx_i) & = \sum_{j,k} - 2 X_{ij}\, (\bp(\bX)\cdot\Delta_{ijk})_+
\nabla_i(\bp(\bX)\cdot\Delta_{ijk})_+ \\
& \quad - \nabla_i(X_{ij}) (\bp(\bX)\cdot\Delta_{ijk})_+^2 \\
 &= -2  \bJ_i \Phi(\bx_i)^T - \sum_{j,k}(\bp(\bX)\cdot\Delta_{ijk})_+^2 \be_{ij},
\end{align*}
where $\be_{ij}$ is the basis vector in the space of $\mC_i$ which is
0 except for a 1 in the direction of $X_{ij}$ and $\bJ_i$ is the 
Jacobian of the price matrix at $\bX$
in the vector space $\mC_i$ given by
$$
\bJ_i = \left[ 
\begin{array}{ccc}
\partial p_1(\bX)/ \partial X_{i\,1} & \partial p_1(\bX)/ \partial X_{i\,2} & \cdots \\
\partial p_2(\bX)/ \partial X_{i\,1} & \partial p_2(\bX)/ \partial X_{i\,2} & \cdots \\
\vdots & \vdots & \ddots
\end{array}
\right].
$$
Hence 
\begin{align}
\nabla_i V(\bx_i) \cdot \Phi(\bx_i) & = -2 \Phi(\bx_i) \bJ_i \Phi(\bx_i)^T 
\label{eqn:nvphi}
\\
\nonumber
& \quad - \sum_{j,k}(\bp(\bX)\cdot\Delta_{ijk})_+^2 
\sum_{l,m} (\bp(\bX)( \be_{ij} \cdot \Delta_{ilm}).
\end{align}

One of the conditions on $\bJ_i$ (the rate of change of
the price)  was that it was monotonically non
decreasing (condition 2 in section \ref{sec:pricing}).  
Hence $ \bX \bJ_i \bX^T \geq 0$ for all $X$.  Thus
$-2 \Phi(\bx_i) \bJ_i \Phi(\bx_i)^T < 0$.

It can also be shown that 
$- \sum_{j,k}(\bp(\bX)\cdot\Delta_{ijk})_+^2 
\sum_{l,m} (\bp(\bX)( \be_{ij} \cdot \Delta_{ilm} \leq 0$ following
\cite{smith1984stability} again
by showing that the remaining term in $\nabla V(\bx_i)\cdot \Phi(\bx_i)$,
given by 
$$
- \sum_{j,k}(\bp(\bX)\cdot\Delta_{ijk})_+^2 
\sum_{l,m} (\bp(\bX)( \be_{ij} \cdot \Delta_{ilm}) < 0.
$$
Since both terms in \eqref{eqn:nvphi} are less than zero
then this gives that $\nabla V(\bx_i)\cdot \Phi(\bx_i) < 0$ for all $i$ 
as required.

\end{document}